\documentclass[journal]{IEEEtran}
\IEEEoverridecommandlockouts
\usepackage{graphicx} 
\usepackage{gensymb}
\usepackage[acronym,nonumberlist]{glossaries}
\makeglossaries
\usepackage{comment}
\usepackage{amsfonts}
\usepackage{bbm}
\usepackage{caption}
\usepackage{subcaption}
\usepackage{graphicx}
\usepackage{multirow}
\usepackage{cite}

\usepackage{amsmath}
\usepackage[ruled,vlined]{algorithm2e}

\newtheorem{proposition}{Proposition}
\usepackage{cite}

\usepackage{booktabs}

\usepackage{tikz}
\usepackage{url}
\usetikzlibrary{shapes.geometric, arrows}

\tikzstyle{startstop} = [rectangle, rounded corners, minimum width=3cm, minimum height=1cm, text centered, draw=black, fill=red!30]
\tikzstyle{process} = [rectangle, minimum width=3cm, minimum height=1cm, text centered, draw=black, fill=orange!30]
\tikzstyle{decision} = [diamond, minimum width=3cm, minimum height=1cm, text centered, draw=black, fill=green!30]
\tikzstyle{arrow} = [thick,->,>=stealth]

\title{Fair and Efficient Scheduling for Sensor Networks via Online Whittle Index Policy}

\author{
\IEEEauthorblockN{
Sokipriala Jonah\IEEEauthorrefmark{1},
Seong Ki Yoo\IEEEauthorrefmark{2},
Saurav Sthapit\IEEEauthorrefmark{3},
Anita Khadka\IEEEauthorrefmark{4}
}

\IEEEauthorblockA{
\IEEEauthorrefmark{1}
Centre for Computational Science and Mathematical Modelling,\\
Coventry University, Coventry, UK
}

\IEEEauthorblockA{
\IEEEauthorrefmark{2}
Wireless Transmission, BT Group, Ipswich, UK
}

\IEEEauthorblockA{
\IEEEauthorrefmark{3}
Software and Security Group, Analog Devices, Edinburgh, UK
}

\IEEEauthorblockA{
\IEEEauthorrefmark{4}
WMG, University of Warwick, UK
}
}

\begin{document}

\newacronym{qos}{QoS}{Quality of Service}
\newacronym{rr}{RR}{Round Robin}
\newacronym{rl}{RL}{Reinforcement Learning}
\newacronym{iot}{IoT}{Internet of thing}
\newacronym{wur}{WUR}{Wake Up Radio}
\newacronym{wus}{WUS}{Wake Up Signal}
\newacronym{pcr}{PCR}{Primary Communication Radio}
\newacronym{ewma}{EWMA}{Exponential Weighted Moving Average}
\newacronym{ack}{ACK}{Acknowledgement}
\newacronym{mse}{MSE}{Mean Squared Error}
\newacronym{rmse}{RMSE}{Root Mean Squared Error}
\newacronym{ucb}{UCB}{Upper Confidence Bound}
\newacronym{aoii}{AoII}{Age of Incorrect Information}
\newacronym{aoi}{AoI}{Age of  Information}
\newacronym{rmab}{RMAB}{Restless Multi-Armed Bandit}
\newacronym{mdp}{MDP}{Markov Decision Process}
\newacronym{mab}{MAB}{Multi-Armed Bandit}
\newacronym{pdr}{PDR}{Packet Delivery Ratio}
\newacronym{snr}{SNR}{Signal to Noise Ratio}
\newacronym{voi}{VoI}{Value of Information}
\newacronym{wsns}{WSNs}{Wireless Sensor Networks} 
\newacronym{lsip}{L-SIP}{Linear Spanish Inquisition Protocol}
\newacronym{fwaoii}{FWAoII}{Fair Whittle index AoII}
\newacronym{waoi}{WAoI}{Whittle index AoI}
\newacronym{waoii}{WAoII}{Whittle index AoII} 
\newacronym{kf}{KF}{Kalman filter}
\newacronym{kfe}{KFE}{KF-Error} 

\newglossaryentry{x}{
    name={\ensuremath{x(k)}},
    description={State representation of \ensuremath{x} at time \ensuremath{k}}
}

\newglossaryentry{xdot}{
    name={\ensuremath{\dot{x}}},
    description={Rate of change of the state \ensuremath{x}}
}

\newglossaryentry{S}{
    name={\ensuremath{\mathbf{S}}},
    description={State matrix}
}

\newglossaryentry{z}{
    name={\ensuremath{z(k)}},
    description={Measurements received by the sink at time \ensuremath{k}}
}

\newglossaryentry{H}{
    name={\ensuremath{H}},
    description={Observation matrix}
}

\newglossaryentry{alpha}{
    name={\ensuremath{\alpha}},
    description={Weighting factor in state estimation}
}

\newglossaryentry{beta}{
    name={\ensuremath{\beta}},
    description={Weighting factor for the rate of change}
}

\newglossaryentry{pdr}{
    name={PDR},
    description={Packet Delivery Ratio}
}

\newglossaryentry{snr}{
    name={SNR},
    description={Signal-to-Noise Ratio}
}

\maketitle
\begin{abstract}
  The \gls{wur} allows resource-constrained, battery-powered sensor nodes to remain in a low-power deep sleep mode while continuously listening for a \gls{wus}.
  The sensor nodes only wake up to transmit packets once the \gls{wus} is received, thereby significantly reducing energy consumption.
However, if the data transmitted by these nodes provides only minimal or no meaningful updates to the remote monitor, the polling process can still lead to wasted energy and increased remote storage demands.
To address this issue, we use \gls{aoii} as a metric to prioritise the polling of nodes that provide valuable updates to the remote monitor.
Polling the best combination of nodes based on \gls{aoii} can be considered as a \gls{rmab} problem, which typically requires knowledge of the monitored process’s transition dynamics; however, these dynamics are often unknown.
To address the unknown dynamics, we propose an online learning process  using state estimation to develop optimal \gls{waoii} and \gls{fwaoii} policies. These policies efficiently poll nodes for relevant updates without relying on assumptions about the underlying transition probabilities.
Our experimental results on both real and synthetic datasets show that the proposed online \gls{waoii} can reduce packet transmissions by up to 70\% compared to the commonly used \gls{rr} technique in \gls{wur} systems, while still keeping \gls{rmse} values within the acceptable error tolerance of the applications.
These findings highlight the significant potential of \gls{waoii} and \gls{fwaoii}  as an effective low-power polling technique for \gls{wur} networks.
\end{abstract}

\glsresetall


\glsresetall
\section{Introduction}
\label{sec:introduction}
In many \gls{wsns}, such as those deployed in smart buildings, industrial monitoring, and environmental sensing, sensor nodes continuously collect and transmit data to a central sink. While this supports real time monitoring, transmitting every observation regardless of its informational value rapidly depletes energy resources, congests the network, and increases the processing load at the sink.
In addition to energy constraints, frequent updates place further strain on network bandwidth and storage capacity at the sink or remote server~\cite{kriouile2023pull,trihinas2017admin,gaura2013edge}. These challenges motivate the development of intelligent transmission strategies that selectively update only when significant state changes occur or when the transmitted information is most valuable.
Given that packet transmission is typically the most energy intensive operation in \gls{wsns}, there is growing interest in goal-oriented communication, where nodes are polled or activated only when their updates provide meaningful information. This shifts the focus from periodic sampling to intelligent polling strategies that aim to conserve energy while maintaining the quality of information at the sink.

A promising metric for designing such strategies is the \gls{aoii}~\cite{tang2024learn,holm2023goal}, which quantifies how outdated or incorrect the sink's estimate of a node’s state becomes over time. Minimising the long-term average \gls{aoii} enables more timely and relevant updates from distributed nodes, while reducing redundant transmissions.
Previous work, as reviewed in Section~\ref{section:background}, has proposed \gls{aoii}-based Whittle index policies under the assumption that the transition dynamics of the observed process are known either to the estimator or to the remote monitor. However, this assumption rarely holds in practical deployments, where system dynamics are often unknown or non-stationary~\cite{liang2024bayesian}. Furthermore, prior approaches frequently overlook fairness in scheduling, leading to scenarios where certain nodes are rarely polled, resulting in outdated or biased system views.

Recent research has explored polling strategies that aim to reduce communication overhead while preserving estimation accuracy. Anay~\textit{et al.}~\cite{deshpande2023energy} proposed a censored Kalman-based polling approach that prioritises sensors with the highest covariance traces, effectively selecting nodes that contribute most to reducing estimation error at the sink.
Saad~\textit{et al.}~\cite{kriouile2023pull} introduced an \gls{aoii}-based polling strategy, where \gls{aoii} serves as a penalty metric that grows as the information at the sink becomes outdated. By polling nodes that minimise the average \gls{aoii}, their approach ensures that the sink is consistently updated with the most relevant information. Building on this, a closed-form Whittle index policy was developed to minimise average \gls{aoii}, under the assumption that the transition probabilities of the observed processes are known.
However, in many real world scenarios particularly in dynamic and heterogeneous wireless environments such transition models are often unavailable or difficult to estimate reliably~\cite{liang2024bayesian}. Additionally, existing methods seldom address the fairness of polling schedules, potentially leaving some nodes unpolled for extended periods, resulting in skewed or incomplete system views.

In this paper, we focus on the general problem of intelligent polling using \gls{aoii} without assuming prior knowledge of system dynamics. We prove that the \gls{aoii}-based polling problem is indexable and admits a threshold based Whittle index policy. Furthermore, we show that the optimal threshold can be learned dynamically using only local state estimates, eliminating the need for predefined transition models. To address fairness, we propose a fairness-aware Whittle index policy that balances polling opportunities across all nodes.
To demonstrate the practicality of our approach, we apply it to a Wake-Up Radio (\gls{wur}) system based on IEEE 802.11ba, where sensor nodes remain in deep sleep until explicitly polled. Standard \gls{wur} implementations typically use round-robin or top-\(k\) polling schemes, which do not account for the value of the information provided by each update. In contrast, our method intelligently prioritises updates based on estimated information value, significantly reducing redundant transmissions while maintaining estimation accuracy an essential requirement for energy-constrained sensing applications.
Polling decisions become especially challenging when a sink monitoring \( N \) nodes is limited to polling only \( M \) at a time based on channel constraint. This problem can be formulated as a \gls{rmab}, which is known to be \textsc{P}-space hard~\cite{hunt1998nc,wang2020restless,wang2023optimistic}. As the number of nodes increases, the complexity of finding the optimal policy grows rapidly, rendering exact solutions infeasible in practice.

To overcome this challenge, we build on the state estimation technique proposed by Gaura~\textit{et al.}~\cite{gaura2013edge}, as detailed in Section~\ref{section:edge_mining}. This enables us to dynamically estimate \gls{aoii} in real time and develop both the online \gls{waoii} and the fairness-enhanced \gls{fwaoii} polling policies without relying on prior knowledge of transition probabilities.
The main contributions of this paper are as follows:

\begin{enumerate}

\item We propose an efficient online Whittle index-based polling policy that leverages state estimates to minimise long-term \gls{aoii}. Notably, our approach operates without requiring any prior knowledge of the underlying system dynamics as detailed in Section \ref{section:edge_mining}.
\item We present rigorous evaluation of the online \gls{waoii} policy on synthetic and actual sensor data from real-world deployments, comparing its performance across key metrics such as polling distribution, energy consumption, and \gls{rmse} against state-of-the-art techniques as detailed in Section \ref{section:experimental_evaluation}.
\item We introduce a fairness constraint into the Whittle index policy, resulting in the \gls{fwaoii} policy, which enhances fairness in transmission opportunities among nodes while minimising the \gls{aoii}. 
This approach addresses the fairness limitations of traditional Whittle index policies, where some arms may remain unpolled as detailed in Section \ref{Section:The_Whittle_index}.
\item Although the Whittle index is known to provide asymptotically optimal results, it is only applicable to problems that satisfy the indexability condition. We prove that the \gls{aoii} for pulling is indexable and admits a threshold-based policy. We also derive the solution for determining the optimal policy \ref{appendix:proof_indexability}.

\item When the optimal penalty threshold is known, the problem simplifies significantly, enabling direct computation of the optimal policy. We further show how this threshold can be dynamically learned in the absence of prior knowledge Algorithm (\ref{alg:DynamicPenaltyUpdate}).
\end{enumerate}
A preliminary version of this work was presented in \cite{jonah2025blackseacom}, which included contributions 1 and 2. In this paper, we extend the previous work by additionally including contributions 3 to 5.

\section{Background and Related Work}
\label{section:background}

\subsection{Intelligent Polling for Resource-Constrained Systems}

Efficient polling mechanisms are essential in \gls{wsns} and \gls{iot} systems, particularly those deployed in energy-constrained environments. While our proposed method is general and applicable to any polling-based system, we use \gls{wur} as a practical case study to demonstrate its effectiveness in reducing transmission overhead and improving information quality. The techniques presented here, however, extend beyond \gls{wur} and are broadly relevant to a wide range of low-power communication systems.

In traditional \gls{wsns}, energy conservation is often achieved through duty cycling, where nodes alternate between active and sleep states. While this reduces idle listening, it introduces latency and delays in data collection~\cite{oller2015has}. Additionally, nodes commonly rely on random access protocols for channel contention, which in dense networks leads to increased collisions, retransmissions, and overall energy consumption.
To address these inefficiencies, the concept of \gls{wur} has gained traction. \gls{wur} introduces an ultra-low power receiver that continuously listens for a \gls{wus}, while keeping the \gls{pcr} in deep sleep~\cite{kozlowski2019energy,deng2020ieee}. When a valid \gls{wus} is detected, the \gls{pcr} is activated to transmit data, thereby reducing idle listening and unnecessary channel contention. This approach significantly lowers the energy cost of monitoring, as the receiver typically operates in the microampere range, achieving up to a 1,000-fold reduction in listening energy~\cite{oller2015has} the \gls{wur} process is described in Algorithm \ref{alg:wake_up_radio_polling}.

Despite these benefits, the energy consumed during transmission remains a critical concern. In particular, transmitting packets that offer little or no informational value to the sink still incurs high energy costs~\cite{yomo2015radio}. This inefficiency motivates a shift toward goal-oriented communication, where updates are transmitted only when they contribute meaningfully to the sink’s understanding of the monitored process~\cite{zhuo2023value,strinati20216g,li2024toward,feng2024goal,getu2023making,talli2024push}.
Several polling strategies have been proposed to improve information relevance. Junya~\textit{et al.}~\cite{shiraishi2020content} introduced a top-\(k\) query-based wake-up scheme for \gls{wur}, where the sink polls nodes with the highest data observations. This method, adapted from the exact top-\(k\) query framework~\cite{malhotra2010exact}, reduces energy consumption by avoiding polling nodes with insignificant updates. However, it assumes that high-magnitude values correspond to high relevance, which is not always the case. For instance, a moderate increase in refrigerator temperature may be more critical than a typical fluctuation in room temperature, even if the latter has a higher numeric value. In such cases, top-\(k\) polling can lead to suboptimal decisions.

These limitations have led to growing interest in information-aware polling strategies that consider the \emph{value} rather than the magnitude of updates. Our work builds on this direction by proposing a general framework based on the \gls{aoii} metric and Whittle index theory, which can be applied across a wide range of resource-constrained systems. While we demonstrate its effectiveness using \gls{wur} as an application case, the proposed method is system-agnostic and suitable for any scenario that requires energy-efficient, intelligent polling.

\begin{algorithm}[ht]
\caption{WUR Polling with Targeted Wake-Up IDs}
\label{alg:wake_up_radio_polling}

\KwIn{Set of sensor nodes $\mathcal{N} = \{n_1, n_2, \ldots, n_N\}$, wake-up IDs $\{ID_1, ID_2, \ldots, ID_N\}$, polling interval $\Delta t$, retry limit $R_{\text{max}}$}
\KwOut{Wake up specific nodes to transmit updates to the sink}

\For{each polling round}{
    Select a node $n_k$ to poll based on the scheduling policy\;
    Sink sends a targeted wake-up signal with $ID_k$ to wake up node $n_k$\;
    Wait for acknowledgment of wake-up from $n_k$\;

    \eIf{wake-up acknowledgment received}{
        Node $n_k$ activates its Primary Communication Radio (PCR) and transmits data\;
        $r \leftarrow 0$\;
        \While{no ACK received from the sink \textbf{and} $r < R_{\text{max}}$}{
            Retry transmission\;
            $r \leftarrow r + 1$\;
        }
        \eIf{ACK received from the sink}{
            Mark transmission as successful\;
        }{
            Mark transmission as failed; deactivate PCR\;
        }
    }{
        Mark wake-up as failed for $n_k$\;
    }
    Wait for $\Delta t$ before the next polling round\;
}
\end{algorithm}

\subsection{aoii Penalty Function}

The \gls{aoii} metric can be regarded as an evolution of the \gls{aoi} concept, initially introduced by Kaul et al.~\cite{kaul2012real} to quantify the freshness of real-time status updates. Since its inception, \gls{aoi} has gained significant traction as a key performance metric in networked monitoring systems. Specifically, \gls{aoi} captures the time elapsed since the most recent successful update was received at the monitor, increasing linearly with time until a new update arrives~\cite{liu2022wireless,jin2022deep,li2020age}. Formally, it is defined as:
\begin{equation}
    \delta \text{AoI}(t) = t - g(t)
\end{equation}
where \(g(t)\) denotes the time at which the last successful update was received at the monitor.
Although \gls{aoi} promotes timely updates, it does not consider whether the underlying state of the observed process has changed. Consequently, updates may be transmitted even when they convey little or no new information, leading to potential inefficiencies in communication~\cite{maatouk2020age}.

Error-based metrics, on the other hand, apply a constant penalty whenever there is a discrepancy between the true state and the estimated state at the monitor, regardless of how long the error persists. These approaches typically focus on minimising the mean squared error \gls{mse}~\cite{ayik2023optimization, chen2021scheduling}. However, they lack temporal sensitivity and fail to prioritise updates based on their urgency or relevance.

To address the limitations of both \gls{aoi} and error-based penalties, the \gls{aoii} penalty function was introduced in~\cite{maatouk2022age}. \gls{aoii} captures the notion of providing fresh and meaningful updates by penalising the system only when the current state of the process deviates from the last reported state. The penalty increases with the degree and duration of this deviation. 
Minimising the penalty \gls{aoii} ensures that updates are not only timely but also relevant prioritising the delivery of state changes that matter. This makes AoII a more robust metric for applications where update value, rather than frequency alone, is essential. The \gls{aoii} penalty function is defined as:

\begin{equation}
     {\text{AoII}}(t) = f(t) \times g(x(t), \hat{x}(t))
\end{equation}
\begin{itemize}
    \item \(f(t)\): A time-dependent penalty function that grows with time and
    represents the cost of not being aware of the correct status of the process at the last updated time. 
    \item \( g(x(t), \hat{x}(t)) \): A function that represents the error-based penalty (e.g, indicator, squared error, threshold error) details can be found in \cite{maatouk2022age}. 
    the error is the difference between the actual state of the process \(x(t)\) and the current estimate \(\hat{x}(t)\) at the monitor.
\end{itemize}
Although \gls{aoii} incorporates both error and penalty, many previous works\cite{ayik2023optimization,maatouk2022age} primarily rely on the indicator or threshold function, which returns a binary value (1 or 0), leading to a uniform penalty. 
However, in most systems, such as fire alarms or industrial process monitoring, the higher the difference between the actual and estimated states, the higher the penalty should be. 
For example, in temperature monitoring, a larger deviation between measured and estimated values should incur a higher penalty. 
This limitation led to the development of the distance-based \gls{aoii} by Kriouile \textit{et al}.~\cite{kriouile2023pull}, where the penalty grows in proportion to the magnitude of the state difference, providing a more accurate representation of the system’s performance. The distance-based \gls{aoii} can be expressed as:
\begin{equation}
\delta \text{AoII}(t) = \mathbb{E} \sum_{u}^{t} \left| x(u) - \hat{x}(t) \right|
\label{eqn:change_in_aoii}
\end{equation}
where \( x(u) \) represents the state value at the time of the last successful update at the sink, and \( \hat{x}(t) \) denotes the current estimate at time \( t \).

Recent studies on polling-based systems using \gls{aoii}~\cite{kriouile2023pull, ayik2023optimization,kriouile2026minimizing,chen2024minimizing} and error-based metrics~\cite{deshpande2023energy} have largely focused on theoretical models and simulations, most of which make the following assumptions:
\begin{itemize}
    \item The state space of the information content is small and known.
    \item When the content of the information changes, there is a high probability that it will revert to its previous state in the near future.
\end{itemize}
As shown in \cite{liu20222}, real-world data collection scenarios often involve extremely large state spaces, making it unlikely for the system to revisit a previously observed state in the near future. Furthermore, most existing approaches remain largely theoretical, offering little guidance on how the \gls{aoii} can be practically estimated without relying on known state transition probabilities.
Previous works in \cite{ayik2023optimization,maatouk2022age,yates2021age} assume that the sink has full knowledge of the process being monitored at the node side. In contrast, we make no such assumption, thereby addressing more realistic scenarios where this knowledge is not available. Our approach performs state estimation at the sink without relying on prior knowledge of the node-side process.
Furthermore, to the best of our knowledge, there has been no practical implementation of state estimation at the sink under such constraints. Most existing works \cite{deshpande2023energy, saurav2023scheduling, peng2022communication} assume that the sink has access to the underlying process dynamics, which may not be feasible in many real-world applications.
Liang \textit{et al.}~\cite{liang2024bayesian} demonstrated that assuming a known probability distribution leads to suboptimal performance when the assumed steady-state probabilities differ from those in realistic scenarios, particularly when the transition dynamics evolve over time. Such assumptions result in inefficient solutions. 
This has prompted recent research into developing online learning policies for \gls{rmab} problems~\cite{liang2024bayesian,wang2023optimistic,  wang2020restless}.

In this work, we adopt a state estimation technique for practical distance-based \gls{aoii}, as detailed in Section \ref{section:edge_mining}, which is subsequently used to formulate our online Whittle index policy. 
Without assuming known state transition probabilities or steady-state process dynamics of the observed process at the sink, we propose a fairness-aware extension to the \gls{aoii} Whittle index policy.
We validate our approach using both real-world and synthetic datasets across a range of simulation scenarios.

\begin{table*}
\centering
\caption{Summary of Notations Used in the Paper}
\label{tab:notation_summary}
\renewcommand{\arraystretch}{1.2}
\setlength{\tabcolsep}{5pt}
\begin{tabular}{l l}
\hline
\textbf{Notation} & \textbf{Description} \\
\hline

\multicolumn{2}{l}{\textit{System and RMAB Parameters}} \\
\hline
$N, M, t, i$ & Number of nodes, number of selected nodes per step, time index, node index \\
$S, s(t)$ & State space and system (joint) state at time $t$ \\
\\

\multicolumn{2}{l}{\textit{State and Estimation}} \\
\hline
$x_i(t), \hat{x}_i(t)$ & True state and estimated state of node $i$ \\
$\mathbf{x}(t) = (x_1(t), x_2(t))^T$ & State vector (value and rate of change) \\
$x_1(t), x_2(t), z(t)$ & Process value, rate of change, and sensor measurement \\
$u, \Delta t$ & Last update time and sampling interval \\
\\

\multicolumn{2}{l}{\textit{Information Freshness Metrics}} \\
\hline
$\delta \mathrm{AoI}(t), \delta \mathrm{AoII}(t)$ & Age of Information and Age of Incorrect Information \\
\\

\multicolumn{2}{l}{\textit{Communication and Actions}} \\
\hline
$a_i(t), \mathbf{a}(t)$ & Node action and joint action vector \\
$\rho, \hat{\rho}$ & Packet Delivery Ratio and its estimate \\
\\

\multicolumn{2}{l}{\textit{Control and Optimisation}} \\
\hline
$C(s(t),a(t)), R(s(t),a(t))$ & Cost and reward functions \\
$\lambda$ & Lagrangian activation (polling) cost \\
$W_i(s_i)$ & Whittle index of node $i$ \\
$S_t^\star, S_t^F$ & Optimal and fairness-aware scheduling sets \\
$\eta, \mathcal{F}_t$ & Fairness window and violating node set \\
$Q_i(s_i,a_i), V(s)$ & Q-value and value function \\
\\

\multicolumn{2}{l}{\textit{State-Space Model and Noise}} \\
\hline
$A, H, v(t), w(t)$ & State transition matrix, observation matrix, process noise, measurement noise \\
\\

\multicolumn{2}{l}{\textit{Estimation Parameters}} \\
\hline
$\beta_1, \beta_2, \beta_3$ & Smoothing parameters for estimation and PDR update \\
\\

\multicolumn{2}{l}{\textit{Energy Model}} \\
\hline
$E_{\max}, E_t, E_s, E_w, E_0$ & Maximum energy, transmission, sensing, wake-up, and sleep energy consumption \\
\\

\hline
\end{tabular}
\end{table*}

\section{Methodology}
\label{section:method}
\subsection{System model }
\begin{figure}
    \centering
    \includegraphics[width=\linewidth]{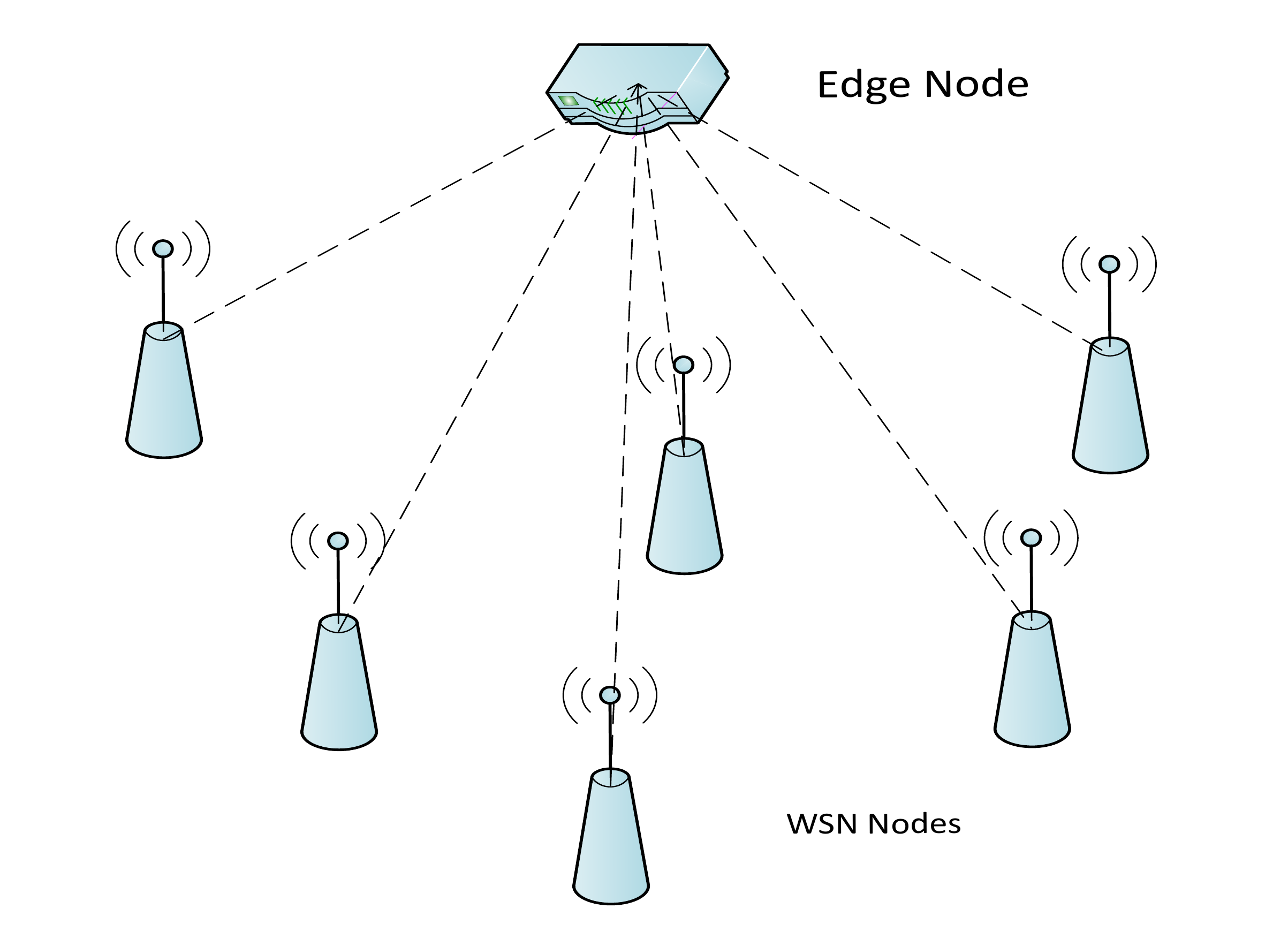}
    \caption{\gls{wsns} with wake-up radio, where the edge node can poll only \(M\) out of \(N\) nodes at each time step.}

    \label{fig:system_model}
\end{figure}
In this work, we consider a scenario in which \(N\) sensor nodes are deployed within a sensing region, such as a smart home, an industrial facility, or an agricultural field. Each sensor node is equipped with a \gls{wur} that continuously listens for \gls{wus} from the sink. Upon receiving a \gls{wus}, the \gls{wur} activates the \gls{pcr}, which then transmits data to the sink over an unreliable channel.
If the transmission fails, a retransmission is initiated. 
If the data is successfully received, an \gls{ack} is sent to the node, and the \gls{pcr} is turned off. Each node measures various processes of interest, \(x_{i}(t)\), which evolve over time. When polled, the node transmits its data to the sink.

Since the sink does not have an accurate knowledge of the monitored process's state, a prediction and estimation technique known as edge mining is used to help the sink construct an estimate of the process, \(\hat{x}_{i}(t)\), based on the data received from the nodes. Further details are provided in Section~\ref{section:edge_mining}.
In this setup, the sink needs to know the status of all nodes. However, due to channel constraints, the sink can only select a subset \( M \) from the \( N \) nodes as illustrated in Fig~\ref{fig:system_model}. Consequently, the sink attempts to poll the \( M \) nodes that will provide the most relevant updates for its state estimation.
The importance of each update is estimated using the \gls{aoii}, which reflects how outdated and inaccurate the last information received at the sink has become.

\subsection{Edge mining technique for state estimation}
\label{section:edge_mining}
The edge mining technique involves directly converting or transforming raw sensory data at the point of collection (the node) into a meaningful format with relevant conceptual information. Central to this process is state estimation, which converts raw sensor data into a meaningful form suitable for the application’s context or the specific goals of the sensing system.
State estimation simplifies event detection by enabling thresholds to be used for monitoring changes in sensor data observations. 
For example, a threshold can be set to detect state transitions, such as identifying a temperature shift from hot to cold. Similarly, thresholds applied to accelerometer data can detect activities like sitting, standing, or walking.

\gls{lsip} is an edge mining technique applied in scenarios where it is desirable to reconstruct the original signal within a specified error bound \cite{gaura2013edge}. 
In \gls{lsip}, nodes only transmit data when the observed parameter deviates significantly enough that the sink can no longer accurately estimate it based on previously received data.
By only transmitting data when there is significant deviation, \gls{lsip} reduces the data transmitted by the node. Similarly, this technique can be applied to directly estimate the \gls{aoii}. 
By polling nodes only when the \gls{aoi} reaches a specified threshold, we can minimise the number of polling requests and, more importantly, determine which nodes to poll at a given time based on how outdated the information at the sink is, as indicated by the \gls{aoii}.

\subsection{Estimation of aoii}

Let us assume \( x(t) \) represents a linear dynamic process:

\begin{equation}
x(t) = A x(t - \Delta t) + v(t)
\end{equation}
\noindent
where \( A \) is the update coefficient for the system, and \( v(t) \) represents the system noise.
The actual measurement model is given by:
\begin{equation}
z(t) = H x(t) + w(t)
\label{eq:state_estimation}
\end{equation}
\noindent
where \( H \) is the observation coefficient that maps the state \( x(t) \) to the measurement \( z(t) \), and \( w(t) \) represents the measurement noise.
To apply the \gls{lsip} technique, we encode the process as a state vector consisting of the current value and the rate of change, represented as \( \mathbf{x} = (x_{1}, x_{2})^{T} \). Using the edge mining technique \gls{lsip} \cite{gaura2013edge}, the sink can estimate the state of the node as \( \hat{\mathbf{x}} = (\hat{x}_{1}, \hat{x}_{2})^{T} \). The edge mining technique is model-agnostic, thus various techniques such as \gls{kf}, normalised least square average or other estimation techniques can be used~\cite{gaura2013edge}.


\begin{algorithm}
\caption{State Estimation from Sink}
\label{alg:sink_estimate}
 \textbf{Update state value} $x_1(t)$:
\[
x_1(t) \leftarrow \beta_1 z(t) + (1 - \beta_1) \left( x_1(t - \Delta t) + x_2(t - \Delta t) \Delta t \right)
\]
 \textbf{Update rate of change} $x_2(t)$:
\[
x_2(t) \leftarrow \beta_2 \frac{x_1(t) - x_1(t - \Delta t)}{\Delta t} + (1 - \beta_2) x_2(t - \Delta t)
\]
 \textbf{Transmit state vector} $\mathbf{x} = (x_1(t), x_2(t))^T$ along with packet information and timestamp $u = g(t)$
 
 \textbf{Sink estimates node state using linear interpolation}:
\[
\begin{bmatrix}
\hat{x_1}(t) \\
\hat{x_2}(t)
\end{bmatrix}
\leftarrow \begin{bmatrix}
1 & t - u \\
0 & 1
\end{bmatrix}
\begin{bmatrix}
x_1(u) \\
x_2(u)
\end{bmatrix}
\]

\end{algorithm}

In the \gls{lsip} technique, the sink operates using state estimates received from the nodes, denoted by \( x_1(u) \) and \( x_2(u) \). Based on these state representations, the sink predicts which node currently has the highest \gls{aoii} by applying linear interpolation, accounting for the time elapsed since the last update.

The parameters \(\beta_1, \beta_2 \in (0, 1)\) in Algorithm~\ref{alg:sink_estimate} are smoothing factors that regulate the influence of current and historical observations, while \( \Delta t \) represents the time interval between consecutive samples. This formulation enables stable and adaptive estimation without requiring explicit knowledge of the underlying system dynamics.

Importantly, this approach does not rely on prior knowledge of transition probabilities or an explicit system model. The only assumption is that the underlying process can be approximated as linear over short time intervals. This assumption is widely valid in many practical sensing applications, including temperature, humidity, and acceleration monitoring, where system dynamics exhibit near-linear behaviour over typical sampling horizons. Such approaches have been successfully applied in real-world sensing systems \cite{gaura2013edge,trihinas2017admin}.

Furthermore, many nonlinear processes encountered in sensing applications can be effectively linearised around an operating point \cite{yuan2017weighted,islam2019linearization}. The proposed framework is model-agnostic, as demonstrated in \cite{gaura2013edge}, provided that a consistent approximation scheme is applied at both the remote monitor and the node devices. As such, the approach can be readily extended to nonlinear systems through standard linearisation techniques. In these settings, \gls{lsip} continues to provide reliable state estimates, particularly in real-time sensing environments where sampling intervals are short and system dynamics can be locally approximated as linear.

\subsection{Evolution of aoii}
Given the process estimates of the nodes at the sink from Algorithm~\ref{alg:sink_estimate}, we can directly obtain the actual values for~(\ref{eqn:change_in_aoii}) as
\begin{equation}
\delta \text{AoII}(t) = (t - u) \left| x_{2}(u) \right|.
\label{eqn:practical_aoii}
\end{equation}
\noindent

From~(\ref{eqn:practical_aoii}), we observe that the \gls{aoii} depends on both the \gls{aoi} and the rate of change of the underlying process. 
Consequently, if the rate of change of the system, \( x_{2}(u) \), is slow during a certain time period, the update at the sink remains relevant for a longer time before becoming outdated, and vice versa.

A higher \( \delta \gls{aoii} \) indicates that the estimate at the sink is more outdated since the last update, making it more valuable to poll this node to obtain a fresher update. 
If the node is successfully polled at the next time step, the sink receives the current estimate, and the \gls{aoii} resets to zero. 
However, due to unreliable link quality, polling a node does not necessarily guarantee a successful update.
The sink also estimates the link quality using the \gls{pdr}, defined as the ratio of successful packet \gls{ack} received to the total number of times the node has been polled within a specified time window, denoted by \( \rho \) \cite{jonah2024adaptive,ansar2017adaptive}. The estimated \gls{pdr} at the sink is updated as follows:
\begin{equation}
    \hat{\rho} = \beta_{3} \rho + (1 - \beta_{3}) \hat{\rho}
    \label{eqn:pdr_estimate}
\end{equation}
where \( \beta_{3} \) \( \in [0, 1] \) is a weighting parameter that adjusts the influence of recent versus older packet updates on the link quality estimation. The state vector \( S \) at the sink for each node is represented as \( (x_{2}(u), u, \rho) \), as these parameters govern the evolution of the \gls{aoii}.

\noindent
\textbf{If the node is not polled} the \gls{aoii} evolves as: 
\begin{equation}
    \delta \text{AoII}(t+1) = (t+1 - u) x_2(u). 
    \label{eqn:aoii_no_polling}
\end{equation}
      
\noindent
\textbf{If the node is polled}, the \gls{aoii} either resets or increases depending on the success or failure of the transmission. The evolution in the next time step is expressed as:

    \begin{equation}
    \delta \text{AoII}(t+1) =
    \begin{cases}
    0, & \text{with prob } \hat{\rho} \\
    (t+1 - u)  x_2(u), & \text{else}
    \end{cases}
    \label{eqn:evolution_of_aoii}
\end{equation}
\noindent
As the state estimate at the sink evolves, regardless of whether a polling action occurs, the polling of nodes can be formulated as a special case of multi-armed bandit techniques, specifically \gls{rmab}, as detailed in Section~\ref{section:section_rmab}.

\subsection{The restless multi-armed bandit problem }
\label{section:section_rmab}
The \gls{rmab} is a generalisation of the \gls{mab} that introduces states for each arm. Unlike stochastic \gls{mab}s, where the reward distribution is static, the reward distribution for each arm in \gls{rmab}s depends on its current state. The term “restless” indicates that each arm’s state evolves over time, regardless of whether it is selected (“polled”) or not. 
Specifically, the state of each arm (representing the node estimate at the sink) evolves differently depending on whether the node is polled, with state transitions governed by one \gls{mdp} when an action is taken and a different \gls{mdp} when no action is taken.
The \gls{rmab} agent must repeatedly select \( M \) out of \( N \) arms (\( M < N \)) at each timestep to maximise the objective function. The total reward at each timestep is the sum of the rewards from all arms, including those not selected. Due to the combinatorial action space and state-dependent rewards, finding the optimal policy for \glspl{rmab} is PSPACE-hard, even in systems with known transition probabilities \cite{wang2020restless,papadimitriou1994complexity}. 
The \gls{rmab} framework has been widely applied to resource allocation problems across various domains, including clinical trials \cite{wang2023optimistic,liang2024bayesian} and wireless networks \cite{ayik2023optimization,mehta2018rested}.
Detailed review of \gls{rmab} can be found in \cite{nino2023markovian}.

Online \gls{rl} methods have been widely adopted for \gls{rmab} problems in settings where the transition dynamics are unknown. These approaches typically rely on learning mechanisms that quantify uncertainty in the system dynamics, after which a Whittle index policy is applied for scheduling decisions. Common techniques include Upper Confidence Bound (UCB)-based exploration strategies \cite{wang2023optimistic,jonah2026adaptive} and Thompson sampling-based approaches \cite{akbarzadeh2023learning,liang2024bayesian}. However, many of these methods assume that a prior distribution over the system parameters is known or can be reasonably specified.

In addition, several works have explored Q-learning-based approaches for \gls{rmab} problems under the average reward criterion \cite{xiong2023finite,biswas2021learn}. Unlike discounted reward formulations, the average reward setting focuses on long-term steady-state performance, which aligns closely with objectives such as minimising the average \gls{aoii}. Despite this, existing approaches are typically designed for scenarios where the system state is directly observable or can be represented with a well-defined state space.

In practical sensing systems, however, the underlying state often needs to be inferred or estimated from observations, introducing additional uncertainty in both state representation and transition dynamics. This makes the direct application of existing \gls{rl}-based \gls{rmab} techniques challenging, particularly in problems where metrics such as \gls{aoii} depend on estimated system states rather than fully observable processes.

Given that the \gls{aoii} from (\ref{eqn:evolution_of_aoii}) evolves even when a node is not polled, and considering the constraint on the number of nodes that can be polled, the optimal polling solution to minimise the \gls{aoii} the problem is a \gls{rmab} problem.
Solving the \gls{rmab} problem typically requires prior knowledge of transition probabilities, as demonstrated in previous work on \gls{aoii} for status updates in wireless networks \cite{ayik2023optimization,kriouile2023pull}. In this work, however, we directly estimate the evolution of the \gls{aoii} using the edge mining technique, as detailed in Algorithm~\ref{alg:sink_estimate}.
To formally define the \gls{rmab} problem, we introduce the \gls{rmab} tuple \( (S, A, C, P) \), where each component is defined as follows:
\begin{itemize}
    \item \textbf{State} (\(S\)): The state \( s \in S \) represents the vector state of each node at the sink at time \( t \).
    
    \item \textbf{Action} (\(A\)): The action \( a \in A \) is a binary decision to poll or not poll a given arm. Specifically, \( a_i(s) = 1 \) if node \( i \) is polled at time \( t \), and \( a_i(s) = 0 \) otherwise.
    
    \item \textbf{Cost} (\(C\)): The cost is modeled as a penalty function based on the \gls{aoii}, represented as \( C(s(t)) = (C_1(s(t)), \dots, C_N(s(t))) \), where \( C_i(s(t)) \) is the penalty for the \gls{aoii} of node \( i \) at time \( t \), depending on  the current state \( s(t) \).
    
    \item \textbf{State Transition} (\(P\)): The exact transition probabilities are unknown, but the next state can be estimated based on the rate of change of the system and the link quality estimate \( \hat{\rho}\) as given in~(\ref{eqn:pdr_estimate}).
\end{itemize}

Our objective is to find a scheduling policy \( \pi \) that minimises the average cost subject to the constraints of available channels \( M \leq N\).
Additionally, to ensure fairness, we require that every node is polled at least once within a given time window \(\eta\).
We now proceed to formulate the optimal scheduling problem.
\begin{equation}
\begin{aligned}
\min_{\pi} \quad & \lim_{T \to \infty} \frac{1}{T} \mathbb{E} \left[ \sum_{t=0}^{T-1} \sum_{i=1}^{N} C_i(s(t), a_i(t)) \right] \\
\text{s.t.} \quad & \sum_{i=1}^{N} a_i(t) \leq M \quad \text{for all } t, \\
                  & \sum_{t=k}^{k+\eta-1} a_i(t) \geq 1 \quad \text{for all } i \text{ and } k.
\end{aligned}
\end{equation}

\noindent
The optimal solution to this problem is known to be intractable \cite{papadimitriou1994complexity}. 
Thus, we relax the constraints to use the Lagrangian as the evaluation metric, as used in related work by Saad and Mohamad~\cite{kriouile2023pull}. 
Thus,
\begin{equation}
\min_{\pi} \frac{1}{T} \mathbb{E} \left[ \sum_{t=0}^{T-1} \sum_{i=1}^{N}  C_i(s(t)) + \lambda \, a_i(s(t)) \right] - \lambda M T
\end{equation}
where \( \lambda \geq 0 \) is the penalty function for scheduling a node, which can be interpreted as adding a penalty \( \lambda \) to the polling action \( a = 1 \). 
Since the term \( \lambda M \) is independent of the policy \( \pi \), it can be eliminated to simplify the analysis.

This relaxation allows us to take actions that do not strictly satisfy the optimal scheduling policy. 
By relaxing the constraint, we can decompose the combinatorial policy into a set of \( N \) independent policies for each node (arm), as demonstrated in \cite{wang2023optimistic,kriouile2024asymptotically}. Consequently, we can drop the index term, reducing the problem to,
\begin{equation}
\min_{\pi} \frac{1}{T} \mathbb{E} \left[ \sum_{t=0}^{T-1} \left( C(s(t)) + \lambda \, a(s(t)) \right) \right]
\end{equation}

\subsection{The Whittle index }
 \label{Section:The_Whittle_index}
The Whittle index, developed by Whittle~\cite{whittle1988restless}, has been widely used in various \gls{rmab} problems to determine the optimal scheduling policy when the \gls{rmab} satisfies the indexability property.

\noindent
\textbf{Definition 1 (Indexability)}  
\textit{An arm is said to be indexable if its passive set under the penalty \( \lambda \) expands monotonically from \( \emptyset \) to the entire state space as \( \lambda \) increases from \( -\infty \) to \( +\infty \). 
Intuitively, this means that if an arm selects a passive action at penalty \( \lambda \) for a given state, it will also select a passive action at any higher penalty \( \lambda_h > \lambda \).}

Consider the value function \( V^{\lambda}(s) \), representing the cost associated with being in state \( s \) when a global penalty \( \lambda \) applies to the active action:
\begin{equation}
    \label{eqn:value_policy}
    V^{\lambda}(s) = C(s) + \lambda,
\end{equation}
\noindent
As \( \lambda \) increases, the penalty makes the active action progressively less favorable, encouraging the passive choice regardless of possible benefits in the next state. Since \( V^{\lambda}(s) \) increases continuously with \( \lambda \) for any given state \( s \), the preference for the passive action also grows, causing the passive set to expand monotonically as \( \lambda \) increases. 
Therefore, if an arm chooses a passive action at a given penalty \( \lambda \), it will also choose a passive action for any higher penalty \( \lambda_h > \lambda \). This behavior satisfies the condition for indexability.
\begin{IEEEproof}
 see Appendix~\ref{appendix:proof_indexability}.
\end{IEEEproof}
\noindent
Since the problem is indexable, the Whittle index policy can be applied to achieve optimal scheduling.
\begin{equation}
W_i(  \hat{\rho}, s_i) = \min \left\{ c : Q_{c}(s_i, 0) = Q_{c}(s_i, 1) \right\}
\label{eqn:whittle_eqn}
\end{equation}
where \( Q_{c}(s, a) \) and \( V_{c}(s) \) are the solutions to the Bellman equation with penalty \( c \) for taking action \( a = 1 \):
\[
Q_{c_i}(s, a) = c a + C(s, a) +  \sum_{s' \in S} P_i(s, a, s') V_{c}(s')
\]
\[
V_{c}(s) = \min_{a \in A} Q_{c}(s, a)
\]
\noindent
The polling action becomes optimal when the Whittle index \( W_i(\hat{\rho}, s_i) \) exceeds the global penalty \( \lambda \), implying that \( c > \lambda \). 
Given this global penalty, the optimal solution to (\ref{eqn:whittle_eqn}) can be obtained by solving the Bellman optimality equation.
\begin{IEEEproof}
 see Appendix~\ref{appendix:proof_threshold}.   
\end{IEEEproof}
However, if the globally optimal penalty is not known, an asymptotically optimal penalty can be determined using Algorithm~\ref{alg:DynamicPenaltyUpdate}.

\begin{algorithm}[ht]
\caption{Dynamic Penalty Update}
\label{alg:DynamicPenaltyUpdate}

\KwIn{Number of nodes $N$, maximum number of nodes to poll $M$}
\KwOut{Updated penalty $\lambda^{(t+1)}$}

Initialise penalty $\lambda^{(0)} \gets 0$\;

\For{each time slot $t \in \mathbb{N}^+$}{
    Compute thresholds $c_i$ for each node $i$\;

    Identify the set $\mathcal{E}$ where $c_i > \lambda^{(t)}$\;
    $\mathcal{E} \gets \{ i \mid c > \lambda^{(t)}, \forall i \in \{1, \ldots, N\} \}$\;

    \eIf{$|\mathcal{E}| \leq M$}{
        $\lambda^{(t+1)} \gets \lambda^{(t)}$\tcp*[r]{No penalty update needed}
    }{
        Sort nodes in $\mathcal{E}$ by $c_i$ in descending order\;
        Identify the $M$-th node $i_M$ in the sorted list\;
        Update penalty $\lambda^{(t+1)} \gets c_{M}$\;
    }

    Use $\lambda^{(t+1)}$ in the polling policy\;
}
\end{algorithm}

\textbf{Definition 2 (Threshold Policy)}  
\textit{A policy is considered a threshold policy if there exists a threshold \( c \) such that the action is passive when \( c < \lambda \) and active when \( c \geq \lambda \).
Thus we  formally define our Whittle index threshold policy \(W^{\lambda}_{\pi}\) based on the penalty.} 
\begin{equation}
W^{\lambda}_{\pi} (\hat{\rho},s) = \left[ \mathbbm{1} ({W_i( \hat{\rho}, s_i) \geq \lambda} ) \right] \forall {i \in N} 
\label{eqn:whittle_index}
\end{equation}
\
Thus, we can use the Whittle index to determine which arms to poll at each time slot, as outlined in Algorithm~\ref{alg:Whittle_index_polling}. 
Our optimal policy is to select the \( M \) arms with the highest index values, unless a constraint is violated for any arm. In such cases, the arm with the violated fairness constraint replaces the last arm in the top \( M \).

\begin{algorithm}[ht]
\caption{Whittle Index for AoII with Fairness}
\label{alg:Whittle_index_polling}

\KwIn{Number of nodes $N$, polling constraint $M$, penalty $\lambda$, fairness threshold $\eta$}

\For{each time slot $t \in \mathbb{N}^+$}{
    Observe AoII values $\{A_i(t)\}_{i=1}^N$ for all nodes\;
    
    Compute Whittle Indices $\{W_i(t)\}_{i=1}^N$ using penalty $\lambda$ and AoII dynamics\;
    
    Select a set $\mathcal{S}_t \subseteq \{1, \ldots, N\}$ such that $|\mathcal{S}_t| = M$ and $W_i(t) \geq W_j(t)$ for all $i \in \mathcal{S}_t$, $j \notin \mathcal{S}_t$\;
    
    \For{each $i \in \mathcal{S}_t$}{
        \If{$t - t_i^{\text{last poll}} > \eta$}{
            Replace $i$ with $j = \arg\min_{k \in \mathcal{S}_t} W_k(t)$ to ensure fairness\;
        }
    }
}
\end{algorithm}

\section{Experimental Evaluation }
\label{section:experimental_evaluation}
We evaluate the performance of our \gls{fwaoii} and \gls{waoii} techniques against several other baseline policies: the standard \gls{rr}, \gls{aoi}-based\cite{pu2023aoi,ramakanth2024monitoring}, and \gls{kf}-based policies\cite{deshpande2023energy,chen2021scheduling}. The standard \gls{rr} policy sequentially selects each node in a cyclic order based on a fixed time schedule. The \gls{aoi}-based policy schedules the node which has the most time since the last successful update was received at the sink. The  \gls{kf} based policy utilises covariance traces to determine which node currently has the highest covariance, thereby selecting that node to minimise  estimation error.

We compare the performance of the proposed technique against baseline methods across different scenarios, using both synthetic data and realistic datasets obtained from actual sensor monitoring conditions. This allows for a comprehensive and rigorous evaluation of our approach relative to existing techniques.

\subsection{Experiments with Synthetic Data} 
We begin by describing the two experimental scenarios used in this study.

\textbf{Scenario One} involves a system of 10 sensors divided into two categories. Each category monitors an environment with different levels of variability. Specifically:

\begin{itemize}
    \item \textbf{Category A:} Sensors $i \in \{1, 2, \dots, 5\}$ monitor an environment with gradual temperature changes.
    \item \textbf{Category B:} Sensors $i \in \{6, 7, \dots, 10\}$ monitor a relatively stable environment.
\end{itemize}

The temperature readings are generated using the following model:
\begin{equation}
    z_i(t) = z_{\text{mean}} + A_i \sin\left(\frac{2\pi t}{P_i}\right) + \mathcal{N}(0, \sigma_i)
    \label{eqn:temp_variation}
\end{equation}

where:
\begin{itemize}
    \item $z_i(t)$ is the temperature reading of sensor $i$ at time $t$,
    \item $A_i$ is the amplitude of variation,
    \item $P_i$ is the period of oscillation,
    \item $\mathcal{N}(0, \sigma_i)$ represents Gaussian noise with zero mean and standard deviation $\sigma_i$.
\end{itemize}

The parameters for each category are as follows:
\begin{itemize}
    \item \textbf{Category A:} $P_i = 500$, $\sigma_i = 0.1$, $A_i = 5$
    \item \textbf{Category B:} $P_i = 500$, $\sigma_i = 0.05$, $A_i = 0$
\end{itemize}

\textbf{Scenario Two} features a more complex system with 30 sensors grouped into three categories. Each category reflects a different rate of change in the monitored system, representative of industrial environments where some machines experience more variability than others.
\begin{figure}[h]
    \centering
    \includegraphics[width=\linewidth]{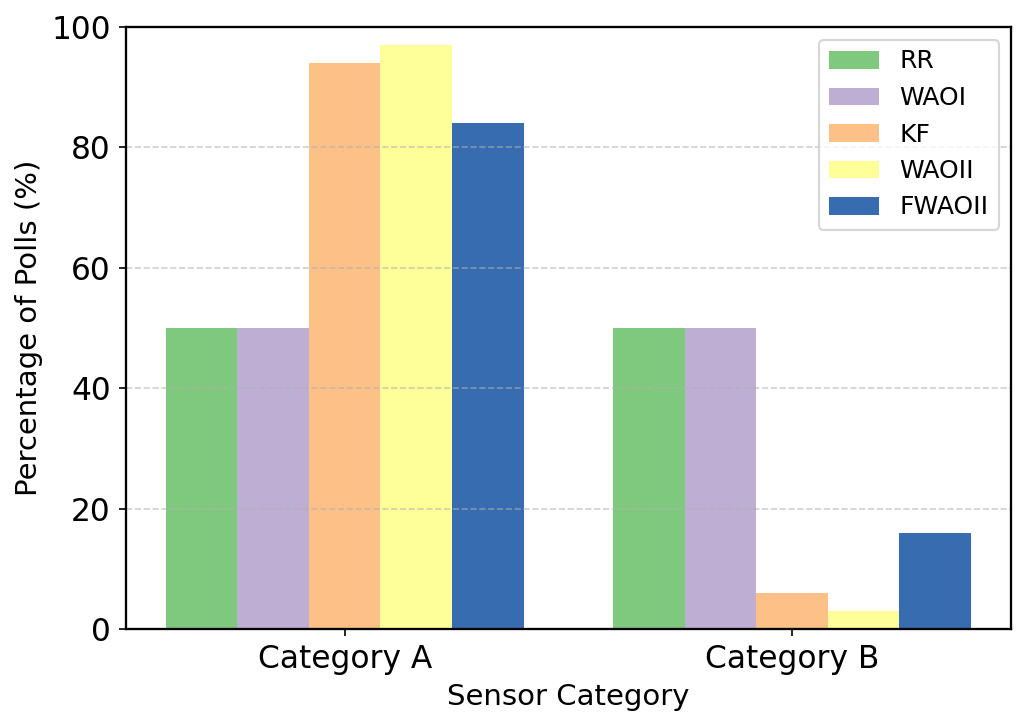}
    \caption{Polling distribution for Scenario One: \gls{rr} and \gls{aoi} techniques sampled sensors from both categories equally, even though sensors in Category B remained relatively stable throughout and offered little change. In contrast, \gls{fwaoii}, \gls{waoii}, and \gls{kf} polled sensors that changed more frequently, providing the sink with fresher updates.}
    \label{fig:pooling_distribution_two_category}
\end{figure}
\noindent

\begin{itemize}
    \item \textbf{Category A:} Sensors $i \in \{1, 2, \dots, 10\}$  slow variation, with $P_i = 1500$, $\sigma_i = 0.05$
    \item \textbf{Category B:} Sensors $i \in \{11, 12, \dots, 20\}$  moderate variation, with $P_i = 1000$, $\sigma_i = 0.05$
    \item \textbf{Category C:} Sensors $i \in \{21, 22, \dots, 30\}$ rapid variation, with $P_i = 500$, $\sigma_i = 0.05$
\end{itemize}

This setup reflects practical scenarios where sensors are deployed across both stable and dynamic environments, allowing us to assess how polling strategies adapt to differences in signal dynamics.

\begin{figure}[h]
    \centering
    \includegraphics[width=\linewidth]{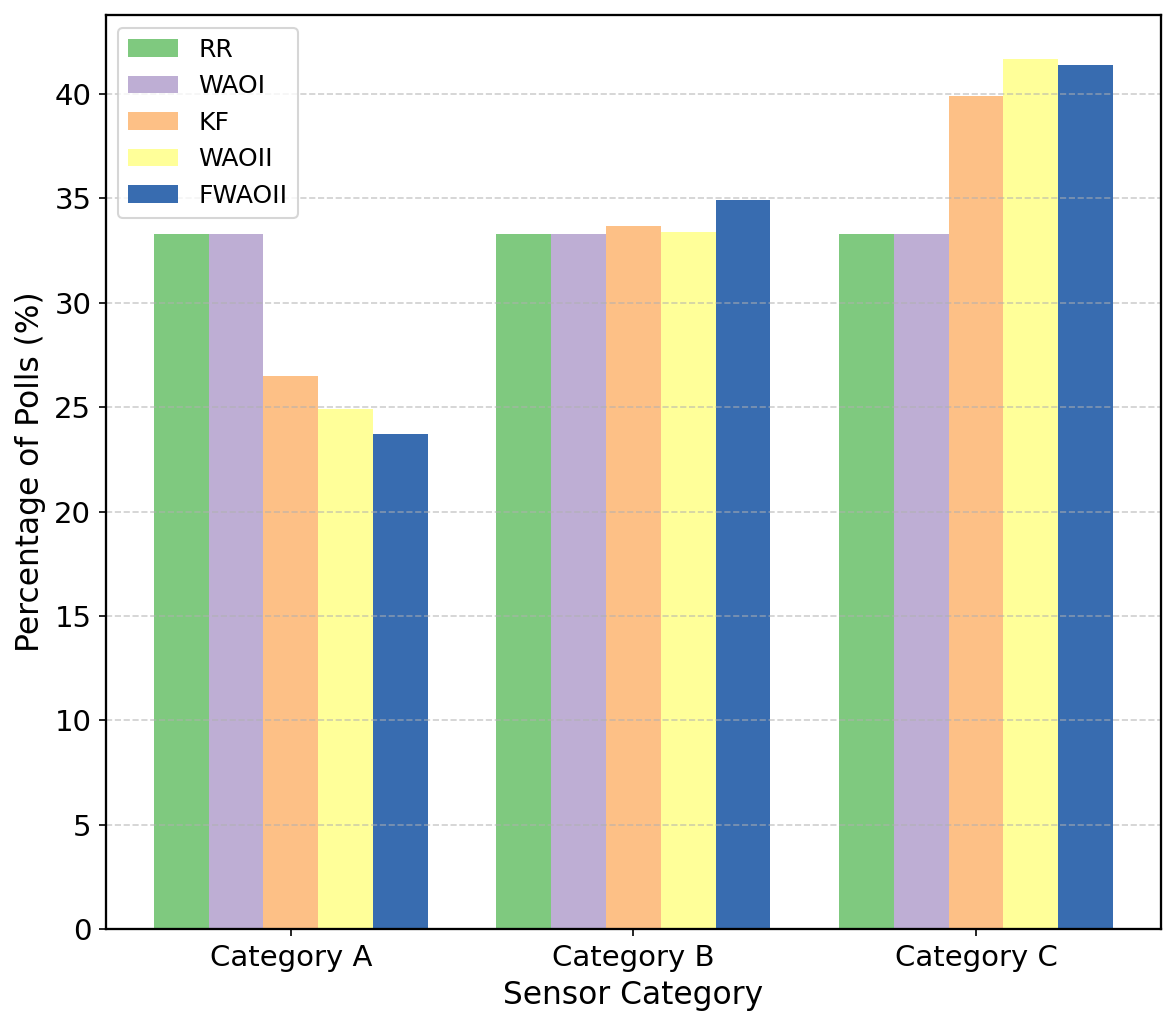}
    \caption{Polling distribution for Scenario Two: The \gls{rr} and \gls{aoi} techniques sampled sensors from all categories equally, despite differences in their rates of change. In contrast, \gls{fwaoii}, \gls{waoii}, and \gls{kf} concentrated sampling on sensors monitoring rapidly changing systems.}
\label{fig:polling_distributition_three_category}
\end{figure}

We analyse how the different techniques sample the nodes. 
In both scenarios, the \gls{rr} and \gls{waoi} based techniques sample each process uniformly, regardless of changes in the underlying system. Consequently, every sensor receives equal scheduling opportunities to transmit, as illustrated in Fig.~\ref{fig:pooling_distribution_two_category} and Fig.~\ref{fig:polling_distributition_three_category}. 
In contrast, the \gls{kf}, \gls{waoii}, and \gls{fwaoii} techniques exhibit adaptive polling, prioritising nodes where the underlying systems vary more significantly. 
This behaviour is particularly evident in scenario one, where over 90\% of polling attempts target Category A sensors with dynamic conditions, reflecting their preference for collecting more valuable updates for the sink.
The performance of the \gls{fwaoii} and \gls{waoii} techniques improves relative to other methods as the channel constraint \(M\) increases, enabling more frequent updates and better estimation accuracy.
We also analyse the impact of the polling penalty \(\lambda_{h}\). 
A higher penalty reduces polling attempts, lowering the number of transmitted packets relative to the \gls{rr} technique, as shown in Table~\ref{table:packet_transmission_penalty}. 
However, reduced sampling leads to increased estimation error. 
The optimal setting is the highest penalty that still maintains an acceptable estimation error.
\begin{table}[h]
    \centering
    \caption{\small Comparison of packet transmission (\% of RR) for varying channel constraint \(M\) .}
\label{table:packet_transmission_M}
   
    \begin{tabular}{c|ccc|cc}
        \hline
        \multirow{2}{*}{\(M\)} & \multicolumn{3}{c|}{Packet Transmission (\% of RR)} & \multicolumn{2}{c}{RMSE} \\
        \cline{2-6}
        & WAoII & KF & WAoI & WAoII & KF \\
        \hline
        \(1\) & 77.28 \% & 83.42\% & 100\% & 0.71 & 0.88 \\
        \(2\) & 40.60 \% & 51.06\% & 100\% & 0.64 & 0.63 \\
        \(5\) & 15.73 \% & 23.32 \% & 100\% & 0.53 & 0.60 \\
        \(10\) & 7.70\% & 11.71\% & 100\% & 0.52 & 0.60 \\
        \hline
    \end{tabular}
     
\end{table}
We also consider a third scenario, where the underlying process varies dynamically over time due to factors such as seasonality. In this case, we extend Scenario One by introducing a reversal in behaviour between the two sensor categories at simulation time \(t = 5000\), such that the roles of the two categories are swapped.
In this scenarios, we primarily highlight the performance of the \gls{fwaoii} and \gls{waoii} techniques. Both methods demonstrate adaptability to changing process characteristics, achieving comparable performance until the midpoint. After this time the \gls{waoii} still does not adapt to the changes as over time the rate of change of the process significantly reduces going to almost zero, while the periodic sampling from the fairness-based \gls{fwaoii} method allows for this thus sampling  all nodes after a fixed time period which enalbe this changes to be noticed and is able to adapt even after fixed behaviour reverse after long duration illustrated in Fig.~\ref{fig:fwaoii_and_waoii}. The response to such changes and also this \gls{fwaoii} also helps detect conditions were a node might stop working, this tradeoff also means that sometimes the optimal action is not taking which increases the number of sampling attempts and the lower the fairness constraints the more sampling done as shown in Table \ref{tab:polling_mse_comparison}. The \gls{fwaoii} despite the less number of packet transmission or scheduling is able to accurately reconstruct the time series of the orignal process as illustrated in Fig~ \ref{fig:spline_reconstruction_synthetic}.
\begin{table}[h]
    \centering
    \caption{\small Comparison of packet transmission (\% of RR) with varying penalty \(\lambda_{h}\) .}
    \begin{tabular}{c|ccc|cc}
        \hline
        \multirow{2}{*}{\(\lambda_{h}\)} & \multicolumn{3}{c|}{Packet Transmission (\% of RR)} & \multicolumn{2}{c}{RMSE} \\
        \cline{2-6}
        & WAoII & KF & WAoI & WAoII & KF \\
        \hline
        \(0.1\) & 18.35\% & 54.30\% & 100\% & 0.4 & 0.32 \\
        \(0.25\) & 16.67\% & 39.45\% & 100\% & 0.44 & 0.51 \\
        \(0.5\) & 15.73 \% & 100\% & 100\% & 0.53 & 0.62 \\
        \hline
    \end{tabular}
    \label{table:packet_transmission_penalty}
\end{table}

\begin{figure}
    \centering
    \includegraphics[width=\linewidth]{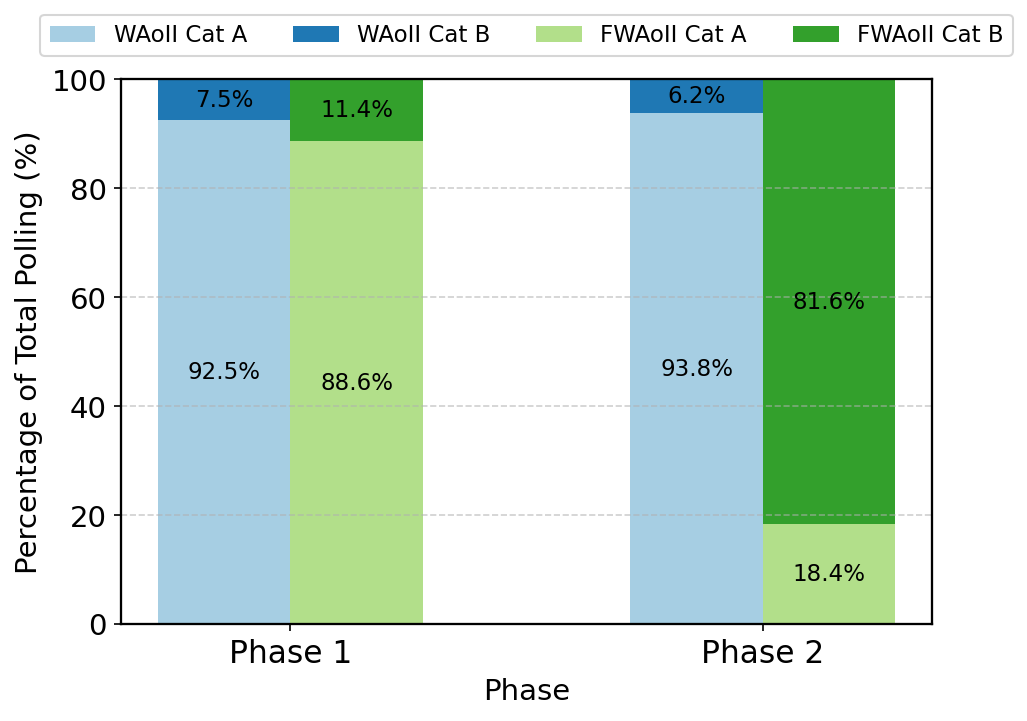}
    \caption{Percentage of polling schedule in Scenario Three. In Phase One, both \gls{fwaoii} and \gls{waoii} prioritise sensors from the more dynamic category, while polling the stable Category B sensors less frequently. After the process reverses at \(t = 5000\), \gls{fwaoii} adapts by shifting its focus to the now-changing category. In contrast, \gls{waoii} continues to deprioritise the previously dynamic category, as its rate of change remains low for an extended period.}
    \label{fig:fwaoii_and_waoii}
\end{figure}

\begin{table}[h]
\centering
\caption{Polling Distribution and RMSE for Different Fairness Windows (\(\eta\))}
\small
\setlength{\tabcolsep}{4pt}
\begin{tabular}{ccccc}
\toprule
\(\boldsymbol{\eta}\) & \textbf{Cat.} & \textbf{Polls} & \textbf{\%} & \textbf{RMSE} \\
\midrule

\multirow{2}{*}{100}
& A & 1243 & 51.64 & \multirow{2}{*}{0.14} \\
& B & 1164 & 48.36 &  \\
\midrule

\multirow{2}{*}{300}
& A & 1190 & 62.02 & \multirow{2}{*}{0.26} \\
& B & 729  & 37.98 &  \\
\midrule

\multirow{2}{*}{500}
& A & 1239 & 87.56 & \multirow{2}{*}{0.60} \\
& B & 176  & 12.44 &  \\
\bottomrule
\end{tabular}
\label{tab:polling_mse_comparison}
\end{table}

\begin{figure}
    \centering
    \includegraphics[width=\linewidth]{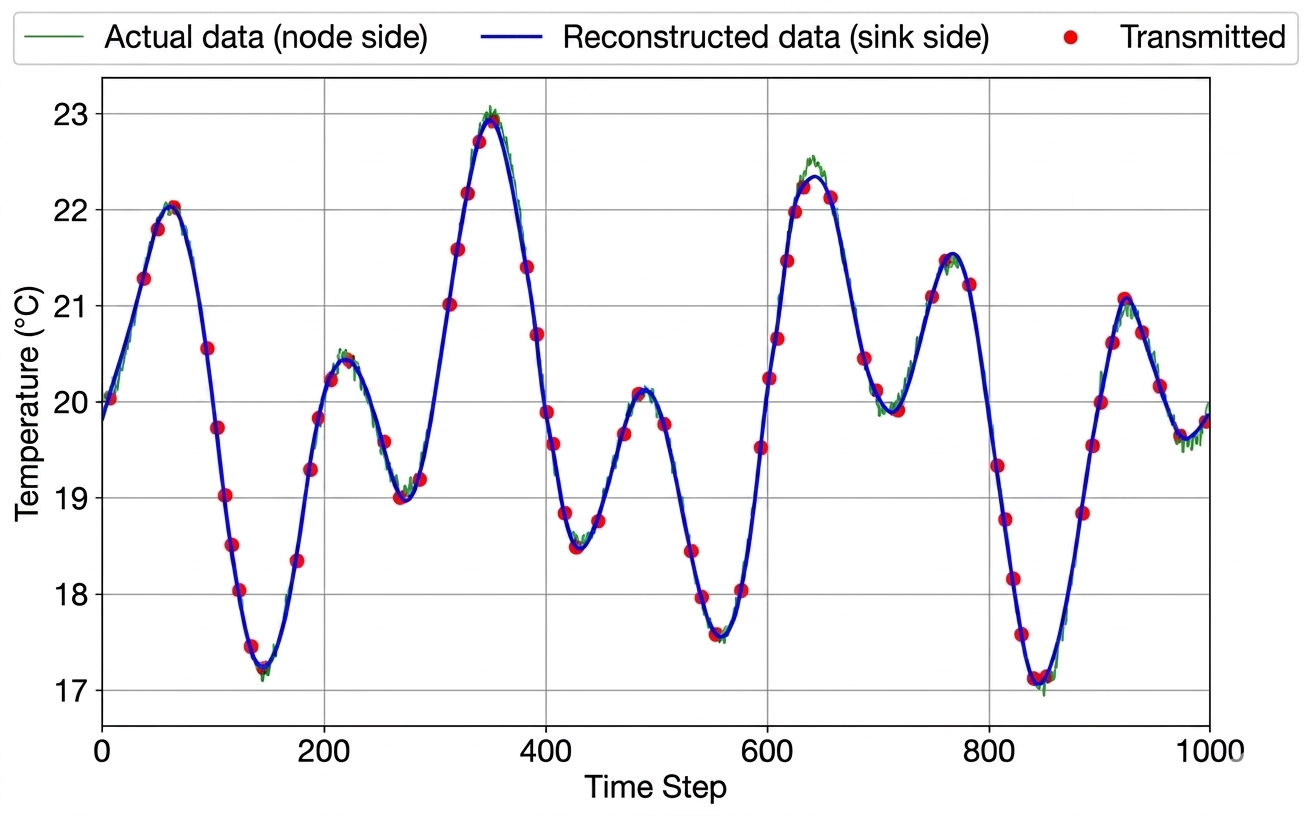}
    \caption{Time series reconstruction from the synthetic data at the sink side}
    \label{fig:spline_reconstruction_synthetic}
\end{figure}

\subsection{Application to Realistic Sensor Data}

\begin{table}[h]
\centering
\caption{Simulation Parameters}
\small
\setlength{\tabcolsep}{4pt}
\begin{tabular}{lcr}
\toprule
\textbf{Parameter} & \textbf{Symbol} & \textbf{Value} \\
\midrule
Nodes                 & \(N\)          & 50 \\
Polls per step        & \(M\)          & \{1,2,5,10\} \\

Transmission energy   & \(E_t\)        & 50 mJ \\
Sensing energy        & \(E_s\)        & 10 mJ \\
Wake-up energy        & \(E_w\)        & 10 mJ \\
Sleep energy          & \(E_0\)        & 1 mJ \\
Battery capacity      & \(E_{\max}\)  & 162 kJ \\
\bottomrule
\end{tabular}
\label{Table:simulation_parameters}
\end{table}
We now focus on sensing data from a wireless sensor deployment, we used the publicly available dataset from Intel Berkeley research lab \cite{madden2010}, which consists of sensor readings collected from 54 sensors deployed within a large building to measure temperature, humidity, voltage, and light intensity the dataset include a log of about 2.3 million readings . 
The sensor readings were collected over a period of four months using the TinyDB in-network query processing system with Mica2Dot sensors \cite{madden2010}.
 
Although the datasets in this work are collected for building monitoring, the outlined technique can be effectively applied to various other sensor-based environments. 
For instance, it could be used in agricultural settings to monitor plant humidity or in industrial facilities to track machine parameters such as temperature, pressure, or voltage. 
These applications involve typical sensor parameters and are well-suited for similar polling-based techniques, making this approach broadly applicable to a range of monitoring scenarios. We compare  the life time duration of \(D_{l}\) using ~\eqref{eqn:life_time} \gls{fwaoii} and \gls{waoii} against other techniques. As shown from Table ~\ref{tab:lifetime_comparison}

\begin{equation}
D_{l} = \frac{1}{N} \sum_{n=1}^{N} \left( \frac{E_{\text{max}}}{\omega_t E_t + \omega_w (E_s + E_w) + (1 - \omega_w) E_0} \right)
\label{eqn:life_time}
\end{equation}
where \(D_{l}\) is the average sensor lifetime, \(N\) is the total number of nodes, \(E_{\text{max}}\) is the maximum energy available, \(\omega_t\) is the transmission frequency, \(E_t\) is the energy consumed during transmission, \(\omega_w\) is the wake-up frequency, \(E_s\) is the energy consumed for sensing, \(E_w\) is the energy consumed for the wake-up process, and \(E_0\) is the energy consumed while in sleep mode.

The reconstruction performance is further evaluated using the sensing data by measuring the \gls{rmse}, in order to assess how accurately the original signals, such as temperature and humidity, can be reconstructed. As illustrated in Fig.~\ref{fig:temp_real_reconstruction} and Fig.~\ref{fig:humidty_real_reconstruction}, spline-based reconstruction is performed using the encoded transformations received at each polling instant, enabling near-accurate signal reconstruction while significantly reducing the number of transmitted packets. The results demonstrate that the proposed approach achieves accurate reconstruction with only a small reconstruction error, while requiring approximately \(90\%\) fewer packet transmissions compared to the standard \gls{rr} approach, as shown in 

\begin{table}[h]
\centering
\caption{Average network lifetime (years) for different channel constraints \(M\) using the temperature dataset (\(\lambda=0.5\)).}
\label{tab:lifetime_comparison}

\small
\setlength{\tabcolsep}{3pt}

\begin{tabular}{cccccc}
\toprule
\multirow{2}{*}{\(M\)} 
& \multicolumn{5}{c}{Network Lifetime (Years)} \\
\cmidrule(lr){2-6}

& \gls{rr} 
& \gls{waoi}
& \gls{waoii} 
& \gls{fwaoii} 
& \gls{fwaoii} \\

& 
& 
& 
& \scriptsize\((\eta=200)\)
& \scriptsize\((\eta=100)\) \\

\midrule
1  & 1.848 & 1.848 & 2.622 & 2.326 & 2.094 \\
2  & 1.127 & 1.127 & 2.494 & 2.274 & 2.012 \\
5  & 0.519 & 0.519 & 2.440 & 2.203 & 1.976 \\
10 & 0.273 & 0.273 & 2.312 & 2.136 & 1.913 \\
\bottomrule
\end{tabular}
\end{table}

Table~\ref{table:life_time_comparison_rmse_real_data}.
\begin{table*}
\centering

\caption{\small Comparison of packet transmission (\% of \gls{rr}) and \gls{rmse} for different techniques across data types.}
\label{table:life_time_comparison_rmse_real_data}
\begin{tabular}{c|ccc|ccc}

\hline
\multirow{2}{*}{Data}
    & \multicolumn{3}{c|}{Packet Transmission (\% of \gls{rr})}
    & \multicolumn{3}{c}{\gls{rmse}} \\
\cline{2-7}
    & \gls{waoii}
    & \begin{tabular}{c} \gls{fwaoii} \\ (\(\eta{=}200\)) \end{tabular}
    & \begin{tabular}{c} \gls{fwaoii} \\ (\(\eta{=}100\)) \end{tabular}
    & \gls{waoii}
    & \begin{tabular}{c} \gls{fwaoii} \\ (\(\eta{=}200\)) \end{tabular}
    & \begin{tabular}{c} \gls{fwaoii} \\ (\(\eta{=}100\)) \end{tabular} \\
\hline
Temperature
    & 12.8\%  & 15.6\%  & 18.3\% 
    & 0.69  & 0.21  & 0.19 \\
Humidity
    & 10.67\% & 11.80\% & 15.80\%
    & 0.82    & 0.70    & 0.70 \\
Light Intensity
    & 10.93\% & 16.16\% & 19.16\%
    & 20.01   & 19.05   & 19.05 \\
\hline
\end{tabular}
\end{table*}
\subsection{Comparison to RL-Based Approaches}

To provide a fair comparison with learning-based \gls{rmab} approaches, we extend our evaluation to include a reinforcement learning (RL)-based Whittle index method. In particular, we adopt the Whittle Index Q-Learning (WIQL) approach proposed in \cite{jonah2026adaptive}, which learns the Whittle index policy under the average reward setting. This allows for a direct comparison between our proposed \gls{waoii}-based policy and existing RL-based solutions within a consistent problem formulation.

To formalise the objective of minimising the \gls{aoii} while accounting for a penalty cost~\( \lambda \) associated with activating an arm, we define the instantaneous cost function as
\begin{equation}
C(s(t), a(t)) = \sum_{i=1}^{N} \text{AoII}_{i}(t) + \lambda \, a_{i}(t),
\label{eq:cost_function}
\end{equation}
where \( \text{AoII}_{i}(t) \) represents the information inaccuracy at time~\( t \) for arm~\( i \), and \( a_{i}(t) \in \{0,1\} \) denotes whether arm~\( i \) is activated. The parameter~\( \lambda \) introduces a trade-off between maintaining information freshness and the cost of activation.

Since the system aims to minimise the total cost, the reward is defined as the negative of the cost function, given by
\begin{equation}
R(s(t), a(t)) = -C(s(t), a(t))
\label{eq:reward_function}
\end{equation}
Under this formulation, maximising the expected cumulative reward is equivalent to minimising the long-term cost of information inaccuracy and transmission. This reformulation enables a consistent comparison with RL-based methods, which are typically expressed in a reward maximisation framework.

\begin{figure*}[h]
    \centering

    \subfloat[\(M = 1\)]{
        \includegraphics[width=0.3\textwidth]{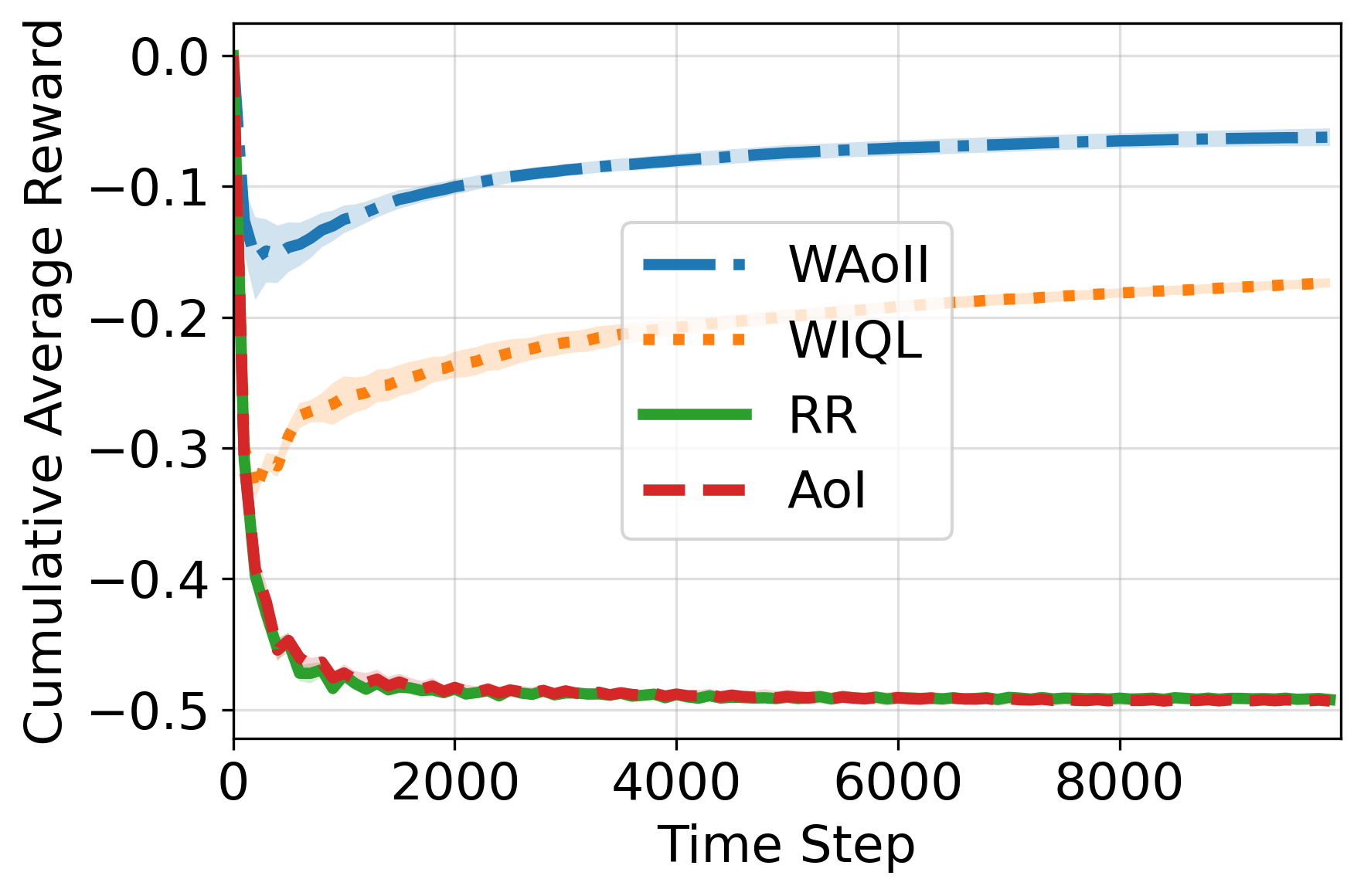}
        \label{fig:m1}
    }
    \hfill
    \subfloat[\(M = 5\)]{
        \includegraphics[width=0.3\textwidth]{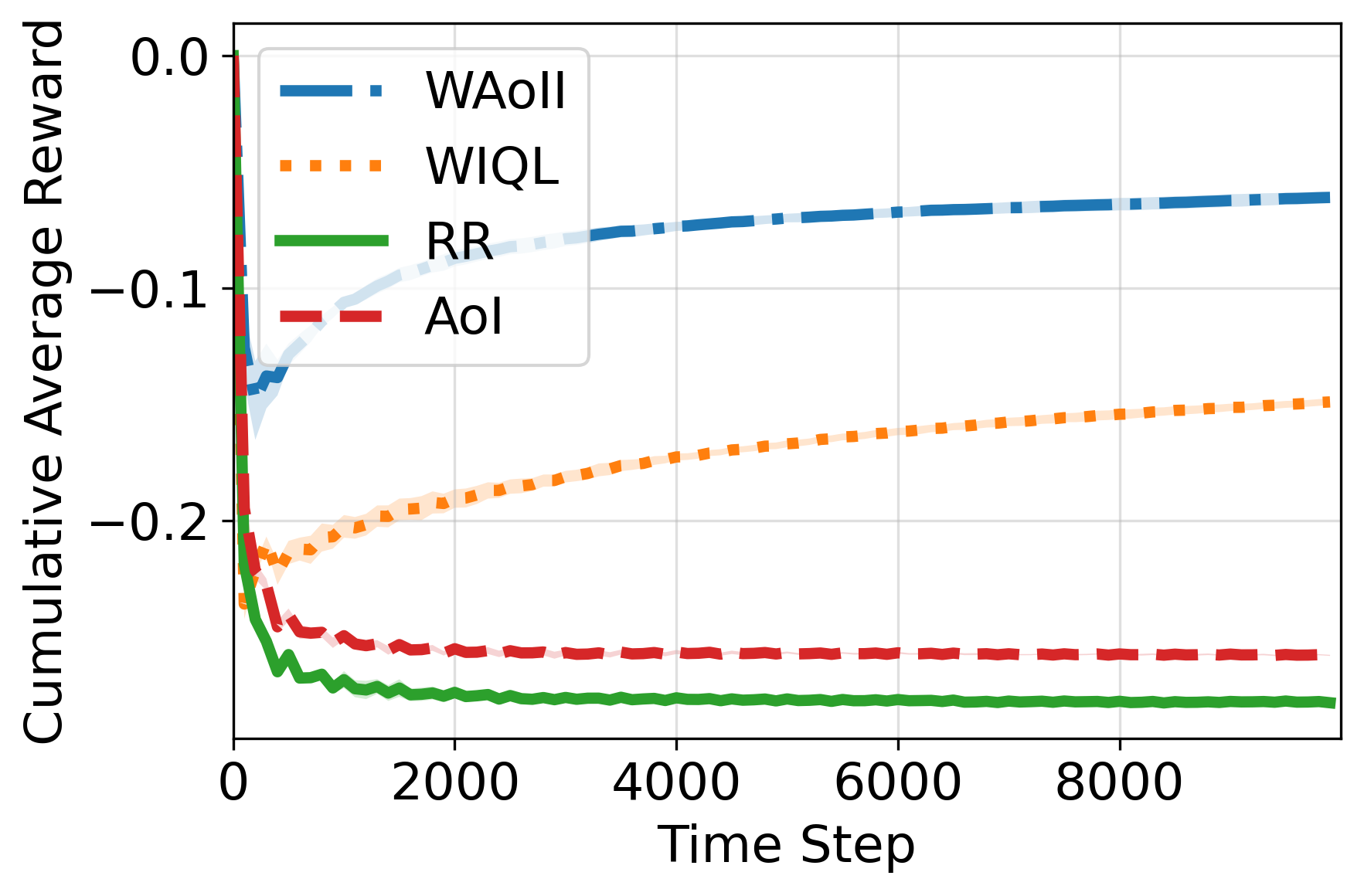}
        \label{fig:m5}
    }
    \hfill
    \subfloat[\(M = 10\)]{
        \includegraphics[width=0.3\textwidth]{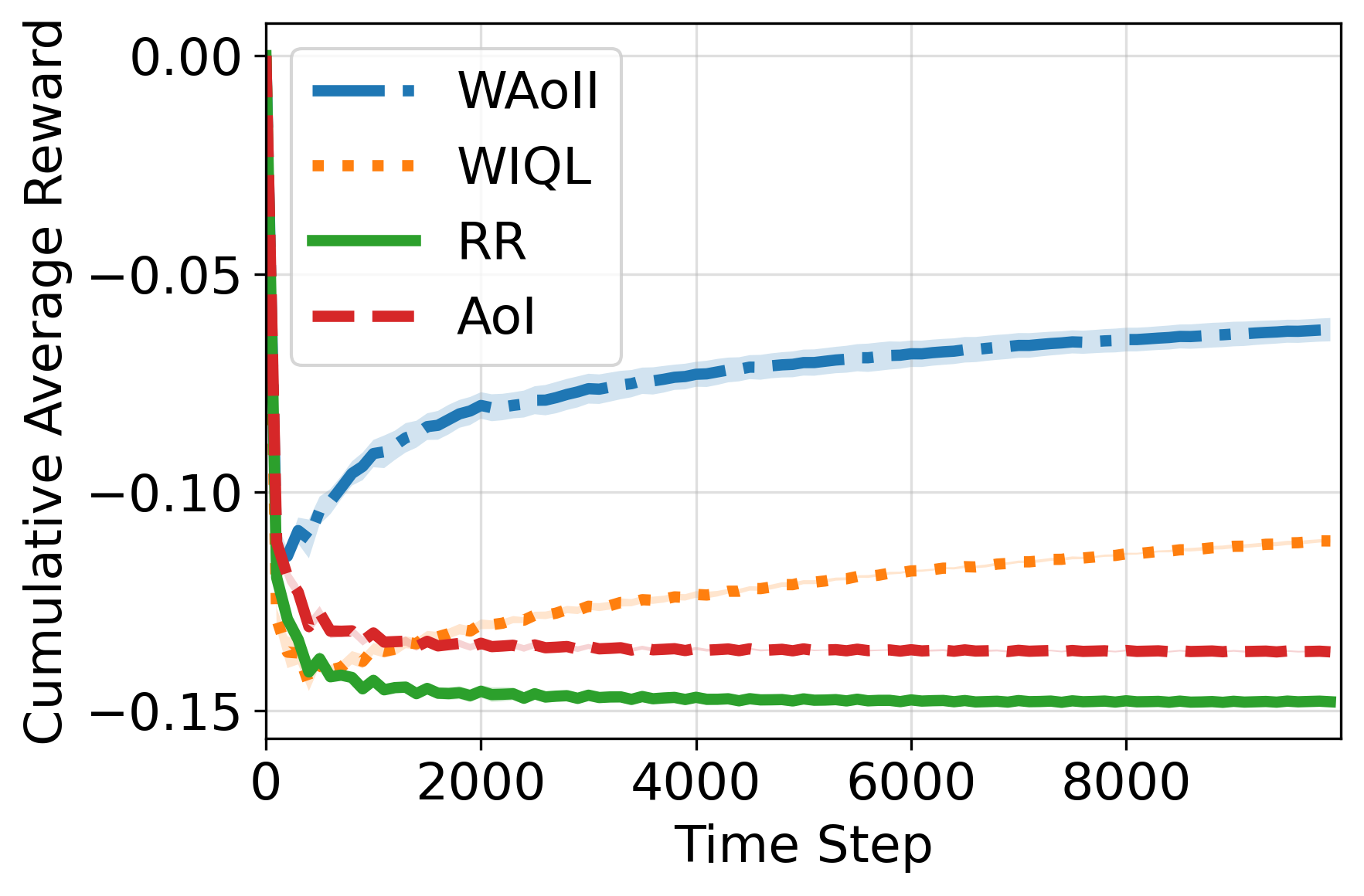}
        \label{fig:m10}
    }

    \caption{Comparison of rewards for various scheduling techniques under different values of \(M\).
     The WAoII approach consistently achieves higher reward compared to WIQL, AoI, and RR policies. The performance gap is most pronounced when the ratio \(N/M\) is high, and decreases as the number of activations \(M\) increases.}
    
    \label{fig:standard_rl_waoii}
\end{figure*}

Fig.~\ref{fig:standard_rl_waoii} presents the comparison of rewards for the proposed WAoII policy against the learning-based WIQL method and classical baselines (AoI and RR) under different scheduling constraints \(M\). The reward is defined as the negative AoII, such that higher reward corresponds to lower information staleness and improved scheduling performance.

Across all settings, the \gls{waoii} policy consistently achieves the highest reward, indicating more effective prioritisation of informative updates. While the WIQL approach demonstrates competitive performance, the performance gap remains relatively small. This suggests that the proposed index-based approach is able to capture most of the gains typically associated with learning-based RMAB strategies, without requiring explicit learning of the index values.

The performance difference is more pronounced when the ratio \(N/M\) is high (e.g., \(M = 1\)), where the scheduling problem is more constrained and efficient prioritisation becomes critical. As \(M\) increases, all methods converge towards similar performance levels due to the increased availability of transmission opportunities, reducing the influence of the scheduling policy.

Importantly, the proposed WAoII approach achieves these results without the need for training, exploration, or hyperparameter tuning, unlike WIQL-based methods. This makes it particularly suitable for real-time and resource-constrained wireless sensor network deployments, where immediate and stable decision-making is required.

\section{Discussion}
In this section, we evaluate the performance of our fairness-based \gls{aoii} Whittle index polling policy using edge mining techniques. 
The evaluation focuses on metrics such as average age, polling distribution, percentage of transmitted packets, and energy consumption.

\begin{figure}
    \centering
    \includegraphics[width=\linewidth]{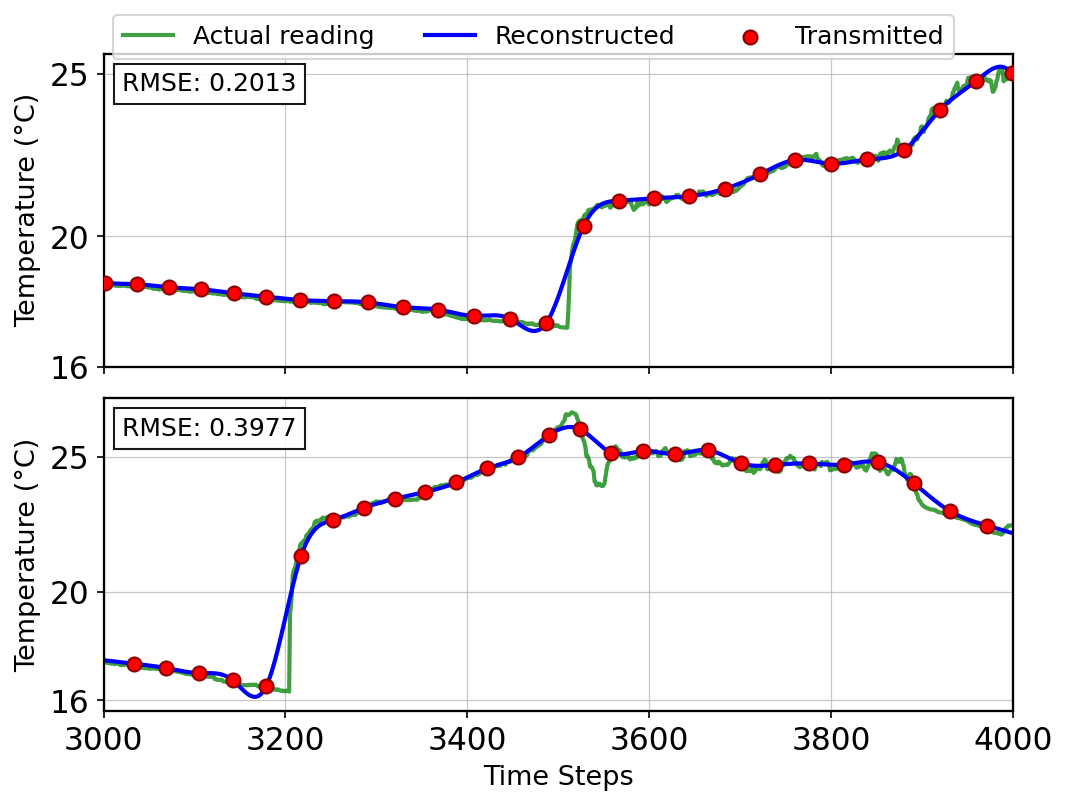}
  \caption{Example time series reconstruction from the temperature real sensor dataset. This demonstrates that with reduced packet transmissions, the \gls{waoii} policy, leveraging state estimation, can accurately reconstruct the original process using the point-in-time value \(x_{1}\) and rate of change \(x_{2}\).}

    \label{fig:temp_real_reconstruction}
\end{figure}

\begin{figure}
    \centering
    \includegraphics[width=\linewidth]{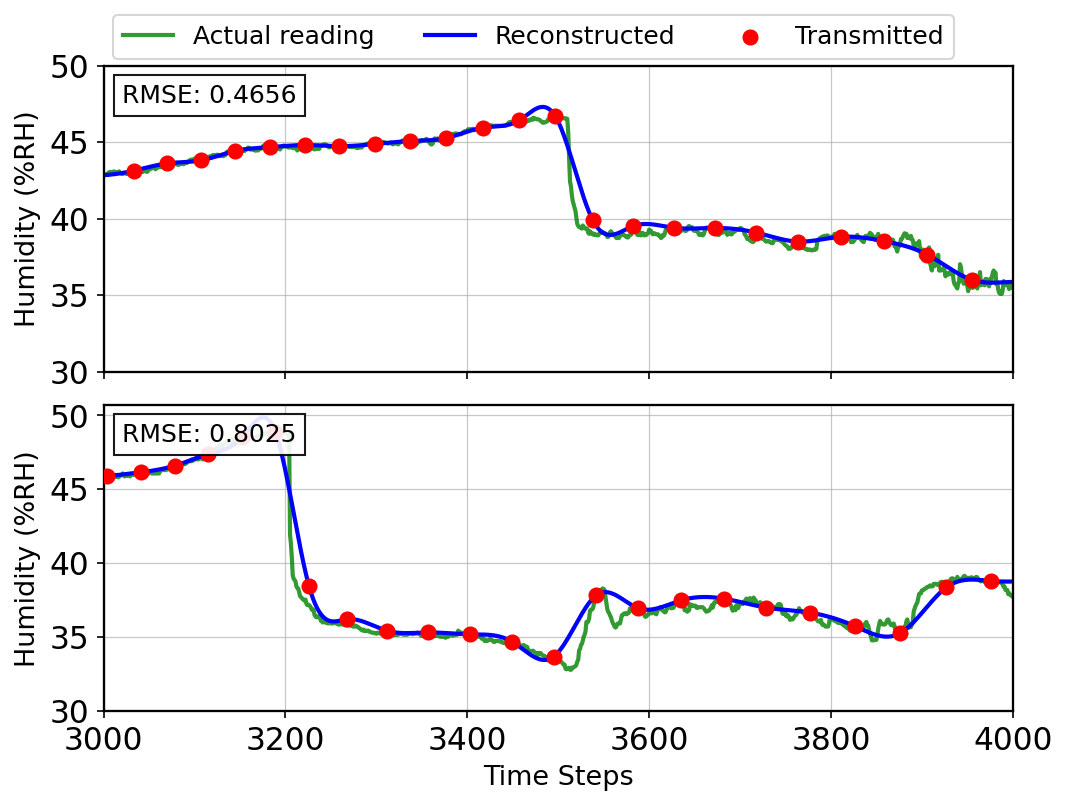}
    \caption{Example time series reconstruction from the humidity real sensor dataset.}
    \label{fig:humidty_real_reconstruction}
\end{figure}
\subsection{Energy consumption}
The \gls{fwaoii} and \gls{waoii} policies significantly outperform \gls{rr} and \gls{waoi} in energy efficiency by transmitting only informative updates based on node-state estimation. Both methods achieve accuracy comparable to \gls{kf} but with substantially reduced packet transmissions (over 90\% fewer), leading to lower energy consumption and storage requirements (see Table~\ref{table:life_time_comparison_rmse_real_data}). Although \gls{waoii} slightly reduces transmissions compared to \gls{fwaoii}, the latter provides improved \gls{rmse} and quicker adaptability to rapidly changing sensor conditions.

\subsection{Impact of Polling Ratio, Fairness Constraint, and Penalty}

The performance of the \gls{fwaoii} policy is evaluated with respect to the polling ratio \(M\), fairness constraint \(\eta\), and penalty \(\lambda\). Increasing \(M\) relative to total nodes \(N\) reduces packet transmissions compared to the \gls{rr} method, though overall transmissions rise with higher \(M\), improving state estimates at the expense of increased energy use (see Table~\ref{table:packet_transmission_M}).

Lower fairness constraints (\(\eta\)) lead to more frequent sampling and better \gls{rmse}, but also increase packet transmissions and energy consumption. The \gls{fwaoii} policy consistently outperforms standard \gls{rr}-based methods for \gls{wur} by prioritising updates that minimise \gls{aoii}, ensuring timely and relevant information. Additionally, edge mining techniques significantly reduce estimation errors compared to raw-data remote estimation approaches.

Our results demonstrate that the \gls{fwaoii} combined with edge mining maintains nearly the same overall information quality at the sink or remote observer, as evidenced by the relatively low \gls{rmse}. 
At the same time, it significantly reduces the number of transmitted packets and prioritises the most relevant updates.

\section{Conclusion and future work}
In this work, we adopt a practical approach to developing \gls{waoii} and \gls{fwaoii} policies that do not rely on assumptions about known probability distributions of system dynamics. These policies prioritise status updates in a \gls{wur}-based polling system constrained by a channel limitation, allowing only a subset of \( M \) out of \( N \) nodes to be sampled at any time. By using state estimation and transforming raw data into meaningful representations, we enable accurate reconstruction of the underlying process at the sink. This approach facilitates the development of an efficient online Whittle index policy.

Our results demonstrate that the \gls{fwaoii} policy dynamically adapts to varying network conditions and sensor readings. Both \gls{waoii} and \gls{fwaoii} achieve acceptable \gls{rmse} across different scenarios, highlighting their robustness for \gls{iot} applications that require reliable data transmission without compromising accuracy. By prioritising nodes with higher \gls{aoii}, the polling policy ensures that the most valuable updates reach the sink, maintaining accurate state estimates with fewer transmissions.
The inclusion of a fairness constraint further guarantees that all nodes have the opportunity to transmit, enhancing adaptability when sensor conditions change after extended periods. The reduction in packet transmissions by the \gls{waoii} and the \gls{fwaoii} results in lower energy consumption, less network congestion, and reduced storage demands, all while delivering timely and critical information.

While this study focused on linear dynamic systems, an important avenue for future work is extending the proposed framework to explicitly handle nonlinear system dynamics. In addition, incorporating uncertainty-aware estimation techniques, such as probabilistic filtering methods, could further improve robustness in noisy or highly dynamic environments. Another key direction is the practical implementation of the proposed scheduling approach on a real-world testbed to evaluate its performance and quantify the actual energy savings achievable in operational sensing systems.

\appendix
\section{\break Proof of indexability}
\label{appendix:proof_indexability}

As the penalty weight \( \lambda \) increases, the cost associated with taking the active action \( V_{\text{active}}^{\lambda}(s) \) also increases. Consequently, the passive action becomes optimal in a greater number of states, especially when the expected benefit of polling does not outweigh its cost.
To formalise this within the \gls{aoii} framework, we define the value functions under penalty \( \lambda \) as follows:
For the active action (i.e., polling), a successful update resets the \gls{aoii} to zero. The immediate cost \( C(s) \) is defined as the current \gls{aoii} value in state \( s \), expressed as:

\begin{equation}
    C(s) = (t - u)  x_2(u)
\end{equation}
where \( (t - u) \) is the time elapsed since the last update, and \( x_2(u) \) is the rate of change of the observed parameter. Since a successful polling resets \gls{aoii} to zero.Assuming a probability \( \hat{\rho} = 1 \), the expected next state value after polling is zero. Thus:
\begin{equation}
  V_{\text{active}}^{\lambda}(s) = C(s) + \lambda = (t - u)  x_2(u) + \lambda.  
\end{equation}

For the passive action (no polling), the \gls{aoii} continues to increase over time. The value of taking the passive action in state \( s \), denoted \( Q(s, \text{passive}) \), includes both the immediate cost and the expected value of the next state. If the node is not polled, then:
\begin{equation}
   Q(s, \text{passive}) = C(s) + \mathbb{E}[V(s') \mid \text{passive}]
\end{equation}
where \( C(s) = (t - u)  x_2(u) \) is the current \gls{aoii}, and \( \mathbb{E}[V(s') \mid \text{passive}] = (t+1 - u) \cdot x_2(u) \) is the expected \gls{aoii} at the next time step. Therefore:
\begin{equation}
   Q(s, \text{passive}) = (t - u)  x_2(u) + (t+1 - u)  x_2(u)
\end{equation}

To determine when the passive action is optimal, we compare the value functions:
\begin{equation}
 V_{\text{active}}^{\lambda}(s) \geq V_{\text{passive}}(s)   
\end{equation}
Substituting the expressions, this inequality becomes:
\[
(t - u)  x_2(u) + \lambda \geq (t - u)  x_2(u) + (t+1 - u)  x_2(u).
\]
Cancelling \( (t - u)  x_2(u) \) from both sides, we find:
\[
\lambda \geq (t+1 - u)  x_2(u).
\]

Define this as a threshold \( c(s) = (t+1 - u)  x_2(u) \) for each state \( s \). Thus, the passive action is optimal if \( \lambda \geq c(s) \).

Define the passive set for a given penalty \( \lambda \) as:
\[
S_{\text{passive}}(\lambda) = \{ s \in S : \lambda \geq c(s) \}.
\]
As \( \lambda \) increases, the inequality \( \lambda \geq c(s) \) holds for more states, meaning that \( S_{\text{passive}}(\lambda) \) expands. Specifically, if \( \lambda < \lambda_h \), then \( S_{\text{passive}}(\lambda) \subseteq S_{\text{passive}}(\lambda_h) \).

Thus, as \( \lambda \) increases from \( -\infty \) to \( +\infty \), the passive set expands monotonically from \( \emptyset \) to the entire state space \( S \), confirming the indexability of the arm. \(\square\)

\subsection{Threshold for optimal polling}
\label{appendix:proof_threshold}
To calculate \( c \), we consider the Bellman equation with a penalty term \( c_i \), where the objective is to find \( c \) such that the Q-value for pulling the arm (i.e., \( a = 1 \)) is equal to the Q-value for not pulling the arm (i.e., \( a = 0 \)). The Bellman equation for the Q-value \( Q_{c_i}(s, a) \) is given by:
\begin{equation}
  Q_{c_i}(s, a) = c a + C(s) +  \sum_{s' \in S} P_i(s, a, s') V_{c}(s')  
\end{equation}
where:
\begin{itemize}
    \item \( Q_{c_i}(s, a) \) is the Q-value for being in state \( s \) and taking action \( a \).
    \item \( C(s) \) is the immediate cost for being in state \( s \) 
    \item \( P_i(s, a, s') \) is the transition probability from state \( s \) to state \( s' \) under action \( a \).
    \item \( V_{c}(s) = \max_{a \in A} Q_{c}(s, a) \) is the value function for state \( s \).
\end{itemize}

To determine when the polling action becomes optimal given a global penalty \( \lambda \), we analyze the \gls{aoii} dynamics under both actions. For the active action (polling), the Q-value can be expressed as:
\begin{equation}
 Q_{c_i}(s, 1) = c + C(s) + \gamma \sum_{s' \in S} P_i(s, 1, s') V_{c}(s') .   
\end{equation}
where \( C(s) \) represents the immediate reward, which depends only on the current state \( s \) and is independent of the action, \( c_i \) is the cost associated with outdated information, and \( \gamma \) is the discount factor. For the passive action (not polling), the Q-value is given by:
\begin{equation}
  Q_{c_i}(s, 0) = C(s) + \sum_{s' \in S} P_i(s, 0, s') V_{c}(s') .  
\end{equation}
Since \( C(s) \) is independent of the action, these terms cancel out in the derivation. Equating \( Q_{c}(s, 1) = Q_{c}(s, 0) \), we derive:
\begin{equation}
  c =  \sum_{s' \in S} \left( P_i(s, 0, s') - P_i(s, 1, s') \right) V_{c}(s')  
\end{equation}
We examine the evolution of AoII under both actions. If the node is not polled, AoII increases linearly with time, described by:
\begin{equation*}
  \delta \text{AoII}(t+1) = (t+1 - u)  x_2(u)  
\end{equation*}
where \( (t+1 - u) \) is the elapsed time since the last update, and \( x_2(u) \) is the rate of change of the observed parameter. 

In the case where the node is polled, the AoII either resets to zero with probability \( \hat{\rho} \) (successful update) or remains unchanged if the transmission fails. As given in Equation~\ref{eqn:evolution_of_aoii}, this dynamic is represented by:
\[
\delta \text{AoII}(t+1) =
\begin{cases}
0, & \text{with probability } \hat{\rho} \\
(t+1 - u)  x_2(u), & \text{otherwise}.
\end{cases}
\]

Assuming a perfect channel (\( \hat{\rho} = 1 \)), each polling action successfully resets AoII to zero. Thus, we have:
\[
c = (t+1 - u)  x_2(u).
\]

If \( x_2(u) \neq 0 \), \( c \) increases over time, eventually surpassing \( \lambda \). This implies that, when the parameter is actively changing (\( x_2(u) > 0 \)), the cost of outdated information will eventually justify the polling action, making it the optimal choice. Therefore, we conclude that under the condition \( x_2(u) \neq 0 \), a threshold \( c_i \) exists beyond which polling is optimal. However, if \( x_2(u) = 0 \), AoII remains constant, and polling may not become optimal.

Whittle Index Policy for Optimal Scheduling

The Whittle index \( W_i(P_i, s_i) \) for each arm \( i \) is determined by solving this threshold condition, allowing for optimal scheduling as follows:
\begin{itemize}
    \item Poll arms where \( W_i(P_i, s_i) > \lambda \),
    \item Leave arms idle where \( W_i(P_i, s_i) \leq \lambda \).
\end{itemize}

This approach provides a systematic method for determining the Whittle index policy, enabling optimal action selection when \gls{aoii} is significant.

\subsection{Complexity Comparison}

We also analyse the computational complexity of the considered scheduling techniques and provide a comparative assessment of their per-step runtime characteristics.

\begin{table}[h]
\centering
\caption{Per-step computational complexity comparison of scheduling techniques.}
\label{tab:policy_complexity_summary}
\renewcommand{\arraystretch}{1.1}
\setlength{\tabcolsep}{6pt}
\begin{tabular}{l c}
\hline
Algorithm & Time per step \\
\hline
Round Robin (RR) & $O(1)$ \\
AoI-based Scheduling & $O(N \log N)$ \\
WAoII (Proposed) & $O(N \log N)$ \\
WIQL (RL-based) & $O(N \log N + N|\mathcal{S}||\mathcal{A}|)$ \\
\hline
\end{tabular}
\end{table}

Table~\ref{tab:policy_complexity_summary} summarises the per-step computational complexity of the evaluated scheduling techniques. The Round Robin (RR) policy has constant time complexity $O(1)$, as it selects the next $M$ nodes in a cyclic order without requiring any state evaluation or prioritisation, making it independent of the total number of nodes $N$.

Both the AoI-based and the proposed WAoII policies incur a complexity of $O(N \log N)$ in the current implementation. This arises from computing a priority index for each node (AoI or AoII-based) followed by sorting to select the top-$M$ nodes for activation. While sorting is used here for simplicity and consistency across methods, more efficient selection strategies (e.g., partial sorting or heap-based selection) could reduce this to $O(N)$ without affecting the scheduling outcome.

In contrast, the WIQL (learning-based) approach introduces additional computational overhead due to the need for value function updates across the state and action spaces. This results in a per-step complexity of $O(N \log N + N|\mathcal{S}||\mathcal{A}|)$, where $|\mathcal{S}|$ and $|\mathcal{A}|$ denote the state and action spaces, respectively. Furthermore, RL-based methods require continuous learning, exploration, and updates, which increases runtime cost and may affect convergence and stability in real-time settings.

Overall, while WAoII and AoI-based approaches have higher computational cost than simple heuristics such as RR, they remain scalable for practical system sizes. Compared to WIQL, the proposed WAoII approach avoids the additional overhead associated with learning and maintains predictable runtime behaviour, making it more suitable for real-time and resource-constrained wireless sensor network deployments.

Furthermore, despite the moderate increase in computational cost relative to RR, the proposed WAoII approach delivers substantial performance gains. Experimental results show that WAoII achieves over $70\%$ reduction in \gls{aoii} compared to both AoI-based and Round Robin policies across multiple scenarios. This demonstrates a favourable performance--complexity trade-off, where similar computational complexity to AoI-based scheduling yields significantly improved information freshness and overall system performance.

\subsection{Effect of the Fairness Constraint on Optimality}

The standard WAoII policy selects, at each time slot, the \(M\) nodes with the largest Whittle indices. Let \(W_i(t)\) denote the Whittle index of node \(i\) at time \(t\), and define the unconstrained scheduling set as
\[
S_t^{\star} = \arg\max_{S \subseteq \{1,\ldots,N\}, |S|=M} \sum_{i \in S} W_i(t).
\]
Thus, \(S_t^{\star}\) is obtained by selecting the top-\(M\) nodes according to their Whittle indices, which maximises the instantaneous index sum.

To enforce fairness, the FWAoII policy introduces a constraint requiring that each node is scheduled at least once within a window of length \(\eta\). Let
\[
\mathcal{F}_t = \{i : t - t_i^{\mathrm{last}} \geq \eta\}
\]
denote the set of nodes that violate the fairness window at time \(t\), where \(t_i^{\mathrm{last}}\) is the most recent scheduling time of node \(i\). The FWAoII policy first selects the top-\(M\) nodes according to the Whittle index and then replaces the lowest-index selected nodes with elements from \(\mathcal{F}_t\), when necessary, to satisfy the fairness constraint.

\begin{proposition}
\label{prop:fairness_loss}
Let \(S_t^{\star}\) be the unconstrained WAoII scheduling set and let \(S_t^{F}\) be the scheduling set selected by the fairness-aware FWAoII policy. If \(r_t = |S_t^{\star} \setminus S_t^{F}|\) nodes are replaced due to the fairness constraint, then the per-step optimality loss of FWAoII relative to WAoII is bounded by
\[
0 \leq 
\sum_{i \in S_t^{\star}} W_i(t)
-
\sum_{i \in S_t^{F}} W_i(t)
\leq
\sum_{k=1}^{r_t}
\left(
W_{j_k}(t) - W_{f_k}(t)
\right),
\]
where \(j_k\) denotes a node removed from the unconstrained top-\(M\) set and \(f_k\) denotes a fairness-constrained node inserted into the schedule.
\end{proposition}

\textit{Proof Sketch.}
The unconstrained WAoII policy maximises the instantaneous sum of Whittle indices by selecting the top-\(M\) nodes. The FWAoII policy deviates from this selection only when fairness violations occur. In such cases, a node \(j_k\) from the selected set is replaced by a fairness-violating node \(f_k\). The resulting loss for each replacement is given by the index difference \(W_{j_k}(t) - W_{f_k}(t)\). Summing over all replacements yields the stated bound. Since \(S_t^{\star}\) maximises the unconstrained objective, the loss is non-negative. \(\square\)

This result shows that the fairness constraint introduces a bounded deviation from the unconstrained Whittle-index policy. The magnitude of this deviation depends on both the number of fairness-triggered replacements and the index gap between the replaced nodes and the inserted nodes. Consequently, FWAoII achieves a controlled trade-off between index-based optimality and fairness in node scheduling.

In the limiting case where \(\eta \rightarrow \infty\), no node violates the fairness window, i.e.,
\[
\mathcal{F}_t = \emptyset,
\]
and therefore no replacements occur. This implies
\[
S_t^{F} = S_t^{\star},
\]
so that FWAoII reduces to the standard WAoII policy and preserves the original Whittle-index optimality. Conversely, smaller values of \(\eta\) enforce fairness more frequently, improving node coverage and adaptability, but potentially increasing deviation from the unconstrained optimum.

\section*{Acknowledgment}
The authors would like to thank Coventry University for the trailblazer  scholarship given to Sokipriala Jonah. This work was completed while Seong Ki Yoo was affiliated with Coventry University. BT, his current affiliation, does not express any opinion on the concepts, conclusions, or recommendations.
\bibliographystyle{IEEEtran}
\bibliography{reference} 

@article{maatouk2022age,
  title={The age of incorrect information: An enabler of semantics-empowered communication},
  author={Maatouk, Ali and Assaad, Mohamad and Ephremides, Anthony},
  journal={IEEE Transactions on Wireless Communications},
  volume={22},
  number={4},
  pages={2621--2635},
  year={2022},
  publisher={IEEE}
}

@article{deshpande2023energy,
  title={Energy-Efficient Internet of Things Monitoring with Content-Based Wake-Up Radio},
  author={Deshpande, Anay Ajit and Chiariotti, Federico and Zanella, Andrea},
  journal={arXiv preprint arXiv:2312.04294},
  year={2023}
}

@article{shiraishi2020content,
  title={Content-based wake-up for top-k query in wireless sensor networks},
  author={Shiraishi, Junya and Yomo, Hiroyuki and Huang, Kaibin and Stefanovi{\'c}, {\v{C}}edomir and Popovski, Petar},
  journal={IEEE Transactions on Green Communications and Networking},
  volume={5},
  number={1},
  pages={362--377},
  year={2020},
  publisher={IEEE}
}

@article{gaura2013edge,
  title={Edge mining the internet of things},
  author={Gaura, Elena I and Brusey, James and Allen, Michael and Wilkins, Ross and Goldsmith, Dan and Rednic, Ramona},
  journal={IEEE Sensors Journal},
  volume={13},
  number={10},
  pages={3816--3825},
  year={2013},
  publisher={IEEE}
}

@inproceedings{kriouile2023pull,
  title={When to Pull Data from Sensors for Minimum Age of Incorrect Information},
  author={Kriouile, Saad and Assaad, Mohamad},
  booktitle={2023 21st International Symposium on Modeling and Optimization in Mobile, Ad Hoc, and Wireless Networks (WiOpt)},
  pages={603--610},
  year={2023},
  organization={IEEE}
}

@article{deng2020ieee,
  title={IEEE 802.11 ba wake-up radio: Performance evaluation and practical designs},
  author={Deng, Der-Jiunn and Lien, Shao-Yu and Lin, Chun-Cheng and Gan, Ming and Chen, Hsing-Chung},
  journal={IEEE Access},
  volume={8},
  pages={141547--141557},
  year={2020},
  publisher={IEEE}
}

@article{oller2015has,
  title={Has time come to switch from duty-cycled MAC protocols to wake-up radio for wireless sensor networks?},
  author={Oller, Joaquim and Demirkol, Ilker and Casademont, Jordi and Paradells, Josep and Gamm, Gerd Ulrich and Reindl, Leonhard},
  journal={IEEE/ACM Transactions on Networking},
  volume={24},
  number={2},
  pages={674--687},
  year={2015},
  publisher={IEEE}
}

@inproceedings{yomo2015radio,
  title={Radio-on-demand sensor and actuator networks (ROD-SAN): System design and field trial},
  author={Yomo, Hiroyuki and Abe, Kenichi and Ezure, Yuichiro and Ito, Tetsuya and Hasegawa, Akio and Ikenaga, Takeshi},
  booktitle={2015 IEEE Global Communications Conference (GLOBECOM)},
  pages={1--6},
  year={2015},
  organization={IEEE}
}

@article{zhuo2023value,
  title={Value of Information-Based Packet Scheduling Scheme for AUV-Assisted UASNs},
  author={Zhuo, Xiaoxiao and Wu, Wen and Tang, Liang and Qu, Fengzhong and Shen, Xuemin},
  journal={IEEE Transactions on Wireless Communications},
  year={2023},
  publisher={IEEE}
}

@article{yuan2017weighted,
  title={Weighted linear dynamic system for feature representation and soft sensor application in nonlinear dynamic industrial processes},
  author={Yuan, Xiaofeng and Wang, Yalin and Yang, Chunhua and Ge, Zhiqiang and Song, Zhihuan and Gui, Weihua},
  journal={IEEE Transactions on Industrial Electronics},
  volume={65},
  number={2},
  pages={1508--1517},
  year={2017},
  publisher={IEEE}
}

@misc{madden2010,
  author       = {Samuel Madden},
  title        = {Intel Lab Data},
  howpublished = {\url{http://db.lcs.mit.edu/labdata/labdata.html}},
  month        = jul,
  year         = 2010,
  note         = {Online; accessed 2010-07-01}
}

@article{kozlowski2019energy,
  title={Energy efficiency trade-off between duty-cycling and wake-up radio techniques in IoT networks},
  author={Koz{\l}owski, Adam and Sosnowski, Janusz},
  journal={Wireless Personal Communications},
  volume={107},
  number={4},
  pages={1951--1971},
  year={2019},
  publisher={Springer}
}

@article{liang2024bayesian,
  title={A bayesian approach to online learning for contextual restless bandits with applications to public health},
  author={Liang, Biyonka and Xu, Lily and Taneja, Aparna and Tambe, Milind and Janson, Lucas},
  journal={arXiv preprint arXiv:2402.04933},
  year={2024}
}

@inproceedings{wang2023optimistic,
  title={Optimistic whittle index policy: Online learning for restless bandits},
  author={Wang, Kai and Xu, Lily and Taneja, Aparna and Tambe, Milind},
  booktitle={Proceedings of the AAAI Conference on Artificial Intelligence},
  volume={37},
  number={8},
  pages={10131--10139},
  year={2023}
}

@inproceedings{ayik2023optimization,
  title={Optimization of aoii and qaoii in multi-user links},
  author={Ayik, Muratcan and Ceran, Elif Tugce and Uysal, Elif},
  booktitle={IEEE INFOCOM 2023-IEEE Conference on Computer Communications Workshops (INFOCOM WKSHPS)},
  pages={1--6},
  year={2023},
  organization={IEEE}
}

@article{strinati20216g,
  title={6G networks: Beyond Shannon towards semantic and goal-oriented communications},
  author={Strinati, Emilio Calvanese and Barbarossa, Sergio},
  journal={Computer Networks},
  volume={190},
  pages={107930},
  year={2021},
  publisher={Elsevier}
}

@article{li2024toward,
  title={Toward Goal-Oriented Semantic Communications: New Metrics, Framework, and Open Challenges},
  author={Li, Aimin and Wu, Shaohua and Meng, Siqi and Lu, Rongxing and Sun, Sumei and Zhang, Qinyu},
  journal={IEEE Wireless Communications},
  year={2024},
  publisher={IEEE}
}

@article{feng2024goal,
  title={Goal-Oriented Wireless Communication Resource Allocation for Cyber-Physical Systems},
  author={Feng, Cheng and Zheng, Kedi and Wang, Yi and Huang, Kaibin and Chen, Qixin},
  journal={IEEE Transactions on Wireless Communications},
  year={2024},
  publisher={IEEE}
}

@inproceedings{kaul2012real,
  title={Real-time status: How often should one update?},
  author={Kaul, Sanjit and Yates, Roy and Gruteser, Marco},
  booktitle={2012 Proceedings IEEE INFOCOM},
  pages={2731--2735},
  year={2012},
  organization={IEEE}
}

@article{yates2021age,
  title={Age of information: An introduction and survey},
  author={Yates, Roy D and Sun, Yin and Brown, D Richard and Kaul, Sanjit K and Modiano, Eytan and Ulukus, Sennur},
  journal={IEEE Journal on Selected Areas in Communications},
  volume={39},
  number={5},
  pages={1183--1210},
  year={2021},
  publisher={IEEE}
}

@article{maatouk2020age,
  title={The age of incorrect information: A new performance metric for status updates},
  author={Maatouk, Ali and Kriouile, Saad and Assaad, Mohamad and Ephremides, Anthony},
  journal={IEEE/ACM Transactions on Networking},
  volume={28},
  number={5},
  pages={2215--2228},
  year={2020},
  publisher={IEEE}
}

@article{kriouile2024asymptotically,
  title={Asymptotically optimal delay-aware scheduling in queueing systems},
  author={Kriouile, Saad and Assaad, Mohamad and Larranaga, Maialen},
  journal={Journal of Communications and Networks},
  year={2024},
  publisher={KICS}
}

@article{malhotra2010exact,
  title={Exact top-k queries in wireless sensor networks},
  author={Malhotra, Baljeet and Nascimento, Mario A and Nikolaidis, Ioanis},
  journal={IEEE Transactions on Knowledge and Data Engineering},
  volume={23},
  number={10},
  pages={1513--1525},
  year={2010},
  publisher={IEEE}
}

@article{wang2020restless,
  title={Restless-UCB, an efficient and low-complexity algorithm for online restless bandits},
  author={Wang, Siwei and Huang, Longbo and Lui, John},
  journal={Advances in Neural Information Processing Systems},
  volume={33},
  pages={11878--11889},
  year={2020}
}

@article{mehta2018rested,
  title={Rested and restless bandits with constrained arms and hidden states: Applications in social networks and 5G networks},
  author={Mehta, Varun and Meshram, Rahul and Kaza, Kesav and Merchant, Shabbir N and Desai, Uday B},
  journal={IEEE Access},
  volume={6},
  pages={56782--56799},
  year={2018},
  publisher={IEEE}
}

@article{chen2021scheduling,
  title={Scheduling to minimize age of incorrect information with imperfect channel state information},
  author={Chen, Yutao and Ephremides, Anthony},
  journal={Entropy},
  volume={23},
  number={12},
  pages={1572},
  year={2021},
  publisher={MDPI}
}

@inproceedings{papadimitriou1994complexity,
  title={The complexity of optimal queueing network control},
  author={Papadimitriou, Christos H and Tsitsiklis, John N},
  booktitle={Proceedings of IEEE 9th annual conference on structure in complexity Theory},
  pages={318--322},
  year={1994},
  organization={IEEE}
}

@article{whittle1988restless,
  title={Restless bandits: Activity allocation in a changing world},
  author={Whittle, Peter},
  journal={Journal of applied probability},
  volume={25},
  number={A},
  pages={287--298},
  year={1988},
  publisher={Cambridge University Press}
}

@article{nino2023markovian,
  title={Markovian restless bandits and index policies: A review},
  author={Ni{\~n}o-Mora, Jos{\'e}},
  journal={Mathematics},
  volume={11},
  number={7},
  pages={1639},
  year={2023},
  publisher={MDPI}
}

@inproceedings{liu20222,
  title={Ao 2 I: Minimizing age of outdated information to improve freshness in data collection},
  author={Liu, Qingyu and Li, Chengzhang and Hou, Y Thomas and Lou, Wenjing and Reed, Jeffrey H and Kompella, Sastry},
  booktitle={IEEE INFOCOM 2022-IEEE Conference on Computer Communications},
  pages={1359--1368},
  year={2022},
  organization={IEEE}
}

@article{saurav2023scheduling,
  title={Scheduling to minimize age of information with multiple sources},
  author={Saurav, Kumar and Vaze, Rahul},
  journal={IEEE Journal on Selected Areas in Information Theory},
  volume={4},
  pages={539--550},
  year={2023},
  publisher={IEEE}
}

@article{peng2022communication,
  title={Communication scheduling by deep reinforcement learning for remote traffic state estimation with Bayesian inference},
  author={Peng, Bile and Xie, Yuhang and Seco-Granados, Gonzalo and Wymeersch, Henk and Jorswieck, Eduard A},
  journal={IEEE Transactions on Vehicular Technology},
  volume={71},
  number={4},
  pages={4287--4300},
  year={2022},
  publisher={IEEE}
}

@article{hunt1998nc,
  title={NC-approximation schemes for NP-and PSPACE-hard problems for geometric graphs},
  author={Hunt III, Harry B and Marathe, Madhav V and Radhakrishnan, Venkatesh and Ravi, Shankar S and Rosenkrantz, Daniel J and Stearns, Richard E},
  journal={Journal of algorithms},
  volume={26},
  number={2},
  pages={238--274},
  year={1998},
  publisher={Elsevier}
}

@article{talli2024push,
  title={Push-and Pull-based Effective Communication in Cyber-Physical Systems},
  author={Talli, Pietro and Mason, Federico and Chiariotti, Federico and Zanella, Andrea},
  journal={arXiv preprint arXiv:2401.10921},
  year={2024}
}

@article{getu2023making,
  title={Making sense of meaning: A survey on metrics for semantic and goal-oriented communication},
  author={Getu, Tilahun M and Kaddoum, Georges and Bennis, Mehdi},
  journal={IEEE Access},
  volume={11},
  pages={45456--45492},
  year={2023},
  publisher={IEEE}
}

@article{liu2022wireless,
  title={Wireless scheduling to optimize age of information based on earliest update time},
  author={Liu, Qingyu and Li, Chengzhang and Hou, Y Thomas and Lou, Wenjing and Reed, Jeffrey H and Kompella, Sastry},
  journal={IEEE Internet of Things Journal},
  volume={10},
  number={7},
  pages={6352--6366},
  year={2022},
  publisher={IEEE}
}

@article{jin2022deep,
  title={Deep reinforcement learning based scheduling for minimizing age of information in wireless powered sensor networks},
  author={Jin, Weiwei and Sun, Juan and Chi, Kaikai and Zhang, Shubin},
  journal={Computer Communications},
  volume={191},
  pages={1--10},
  year={2022},
  publisher={Elsevier}
}

@article{li2020age,
  title={Age-of-information aware scheduling for edge-assisted industrial wireless networks},
  author={Li, Mingyan and Chen, Cailian and Wu, Huaqing and Guan, Xinping and Shen, Xuemin},
  journal={IEEE Transactions on Industrial Informatics},
  volume={17},
  number={8},
  pages={5562--5571},
  year={2020},
  publisher={IEEE}
}

@inproceedings{jonah2024adaptive,
  title={Adaptive Retransmission for Wireless Sensor Nodes Under Bursty Error Conditions},
  author={Jonah, Sokipriala and Yoo, Seong Ki and Sthapit, Saurav},
  booktitle={2024 5th International Conference on Smart Sensors and Application (ICSSA)},
  pages={1--6},
  year={2024},
  organization={IEEE}
}

@inproceedings{ansar2017adaptive,
  title={Adaptive burst transmission scheme for WSNs},
  author={Ansar, Zeeshan and Dargie, Waltenegus},
  booktitle={2017 26th International Conference on Computer Communication and Networks (ICCCN)},
  pages={1--7},
  year={2017},
  organization={IEEE}
}

@inproceedings{trihinas2017admin,
  title={ADMin: Adaptive monitoring dissemination for the Internet of Things},
  author={Trihinas, Demetris and Pallis, George and Dikaiakos, Marios D},
  booktitle={IEEE INFOCOM 2017-IEEE conference on computer communications},
  pages={1--9},
  year={2017},
  organization={IEEE}
}

@article{pu2023aoi,
  title={Aoi-bounded scheduling for industrial wireless sensor networks},
  author={Pu, Chenggen and Yang, Han and Wang, Ping and Dong, Changjie},
  journal={Electronics},
  volume={12},
  number={6},
  pages={1499},
  year={2023},
  publisher={MDPI}
}

@article{ramakanth2024monitoring,
  title={Monitoring correlated sources: Aoi-based scheduling is nearly optimal},
  author={Ramakanth, R Vallabh and Tripathi, Vishrant and Modiano, Eytan},
  journal={IEEE Transactions on Mobile Computing},
  year={2024},
  publisher={IEEE}
}

@unpublished{jonah2025blackseacom,
  author    = {Sokipriala Jonah and Seong Ki Yoo and Saurav Sthapit},
  title     = {Energy Efficient Wake Up Radio Polling Based on Value of Information},
  note      = {Presented at the IEEE International Black Sea Conference on Communications and Networking (BlackSeaCom), Chisinau, Moldova, 23--26 June 2025},
  year      = {2025}
}

@article{tang2024learn,
  title={Learn to Schedule: Data Freshness-Oriented Intelligent Scheduling in Industrial IoT},
  author={Tang, Jianhua and Chen, Fangfang and Li, Jiaping and Liu, Zilong},
  journal={IEEE Transactions on Cognitive Communications and Networking},
  year={2024},
  publisher={IEEE}
}

@article{holm2023goal,
  title={Goal-oriented scheduling in sensor networks with application timing awareness},
  author={Holm, Josefine and Chiariotti, Federico and Kal{\o}r, Anders E and Soret, Beatriz and Pedersen, Torben Bach and Popovski, Petar},
  journal={IEEE Transactions on Communications},
  volume={71},
  number={8},
  pages={4513--4527},
  year={2023},
  publisher={IEEE}
}

@article{jonah2026adaptive,
  title={Adaptive Scheduling: A Reinforcement Learning Whittle Index Approach for Wireless Sensor Networks},
  author={Jonah, Sokipriala and Yoo, Seong Ki and Sthapit, Saurav},
  journal={IEEE Access},
  year={2026},
  publisher={IEEE}
}

@article{akbarzadeh2023learning,
  title={On learning Whittle index policy for restless bandits with scalable regret},
  author={Akbarzadeh, Nima and Mahajan, Aditya},
  journal={IEEE Transactions on Control of Network Systems},
  volume={11},
  number={3},
  pages={1190--1202},
  year={2023},
  publisher={IEEE}
}

@article{xiong2023finite,
  title={Finite-time analysis of whittle index based Q-learning for restless multi-armed bandits with neural network function approximation},
  author={Xiong, Guojun and Li, Jian},
  journal={Advances in Neural Information Processing Systems},
  volume={36},
  pages={29048--29073},
  year={2023}
}

@article{biswas2021learn,
  title={Learn to intervene: An adaptive learning policy for restless bandits in application to preventive healthcare},
  author={Biswas, Arpita and Aggarwal, Gaurav and Varakantham, Pradeep and Tambe, Milind},
  journal={arXiv preprint arXiv:2105.07965},
  year={2021}
}

@article{islam2019linearization,
  title={Linearization of the sensors characteristics: A review},
  author={Islam, Tarikul and Mukhopadhyay, SC},
  journal={International Journal on Smart Sensing and Intelligent Systems},
  volume={12},
  number={1},
  pages={1--21},
  year={2019},
  publisher={Massey University}
}

@article{kriouile2026minimizing,
  title={Minimizing the age of incorrect information for unknown Markovian source},
  author={Kriouile, Saad and Assaad, Mohamad},
  journal={IEEE Transactions on Networking},
  year={2026},
  publisher={IEEE}
}

@article{chen2024minimizing,
  title={Minimizing age of incorrect information over a channel with random delay},
  author={Chen, Yutao and Ephremides, Anthony},
  journal={IEEE/ACM Transactions on Networking},
  volume={32},
  number={4},
  pages={2752--2764},
  year={2024},
  publisher={IEEE}
}

\end{document}